\documentclass{article}

\usepackage{graphicx} 
\def\withcolors{1}
\def\withnotes{1}
\usepackage{ccanonne}
\usepackage{multicol}
\usepackage{multirow}
\usepackage{subcaption}
\usepackage{bbm}
\usepackage{braket}
\usepackage{algorithm}
\usepackage{algpseudocode}
\usepackage{turnstile}

\newtheorem{fact}[theorem]{Fact}
\allowdisplaybreaks
\newcolumntype{C}{>{\centering\arraybackslash}p{0.35\textwidth}}
\newcommand{\dims}{d}

\newcommand{\cD}{{\mathcal{D}}}
\newcommand{\cV}{{\mathcal{V}}}

\newcommand{\povmset}{{\mathfrak{M}}}
\newcommand{\nqubits}{{N}}

\newcommand{\opnorm}[1]{{\left\|#1\right\|}_{\text{op}}}
\newcommand{\tracenorm}[1]{{\left\|#1\right\|}_{1}}
\newcommand{\hsnorm}[1]{{\left\|#1\right\|}_{\text{HS}}}
\newcommand{\barDelta}{{\overline{\Delta}}}
\newcommand{\ptb}{{z}}
\newcommand{\ptbDistr}{{\mathcal{D}_{\ell,\cd}}}

\newcommand{\cd}{{c}}

\newcommand{\isthestate}{\texttt{YES}}
\newcommand{\notthestate}{\texttt{NO}}

\newcommand{\cA}{\mathcal{A}}
\newcommand{\cB}{\mathcal{B}}
\newcommand{\cC}{\mathcal{C}}

\newcommand{\bfP}{\mathbf{P}}

\newcommand{\x}{\mathbf{x}}
\newcommand{\out}{{x}}

\def\multiset#1#2{\ensuremath{\left(\kern-.3em\left(\genfrac{}{}{0pt}{}{#1}{#2}\right)\kern-.3em\right)}}

\newcommand{\ham}[2]{\operatorname{d}_{\rm Ham}\left(#1,#2\right)}
\newcommand{\EMD}[2]{\operatorname{d}_{\rm EM}\left(#1,#2\right)}

\newcommand{\variance}[2]{\var_{#1}{\mleft[#2\mright]}}

\newcommand{\qmm}{{\rho_{\text{mm}}}} 
\newcommand{\qkn}{{\rho_0}} 

\newcommand{\HH}{\mathbb{H}}
\newcommand{\Herm}[1]{{\HH_{#1}}}

\newcommand{\qbit}[1]{|{#1}\rangle}
\newcommand{\qadjoint}[1]{\langle{#1}|}
\newcommand{\qproj}[1]{\qbit{#1}\qadjoint{#1}}
\newcommand{\qoutprod}[2]{\qbit{#1}\qadjoint{#2}}

\newcommand{\qdotprod}[2]{\langle#1|#2\rangle}
\newcommand{\hdotprod}[2]{\left\langle#1,#2\right\rangle}
\newcommand{\matdotprod}[3]{\langle#1|#2|#3\rangle}

\newcommand{\eye}{\mathbb{I}}

\newcommand{\img}{\text{i}}

\newcommand{\rk}{{r}}
\newcommand{\VecOp}{\text{vec}}

\newcommand{\vvec}[1]{|#1\rangle\rangle}
\newcommand{\vadj}[1]{\langle\langle#1|}
\newcommand{\vvdotprod}[2]{\left\langle\left\langle#1|#2\right\rangle\right\rangle}

\newcommand{\bx}{\mathbf{x}}
\newcommand{\outset}{{\mathcal{X}}}

\newcommand{\Luders}{\mathcal{H}}

\newcommand{\Choi}{{\mathcal{C}}}

\newcommand{\hbasis}{{\mathcal{V}}}
\newcommand{\qest}{{\hat{\rho}}}

\newcommand{\dm}{\mathrm{d}}
\newcommand{\Sp}{\mathbb{S}}

\newcommand{\Sim}{\mathcal{S}}

\newcommand{\Haar}[1]{{\mathcal{U}_{#1}}}
\newcommand{\POVM}{\mathcal{M}}

\newcommand{\cycle}{\mathcal{C}}

\newcommand{\constr}{{\mathcal{A}}}

\title{Adversarially robust quantum state learning and testing}

\author{
    \begin{tabular}[t]{C@{\extracolsep{2em}} C}
   Maryam Aliakbarpour\thanks{Computer science department and Ken Kennedy Institute.} &Vladimir Braverman\thanks{Part of the work was completed at Rice University.} \\
 Rice University & Johns Hopkins University\\ 
\small \texttt{maryama@rice.edu} &\small \texttt{vova@cs.jhu.edu} 
\end{tabular}
\vspace{2ex}\\
\begin{tabular}[t]{C@{\extracolsep{2em}} C}
    Nai-Hui Chia & Yuhan Liu \\
Rice University & Rice University\\ 
\small \texttt{nc67@rice.edu} &\small \texttt{yuhan-liu@rice.edu} 
\end{tabular}
}

\begin{document}
\maketitle
\begin{abstract}
Quantum state learning is a fundamental problem in physics and computer science. As near-term quantum devices are error-prone, it is important to design error-resistant algorithms. Apart from device errors, other unexpected factors could also affect the algorithm, such as careless human read-out error, or even a malicious hacker deliberately altering the measurement results. Thus, we want our algorithm to work even in the worst case when things go against our favor. 

We consider the practical setting of single-copy measurements and propose the $\gamma$-adversarial corruption model where an imaginary adversary can arbitrarily change $\gamma$-fraction of the measurement outcomes. This is stronger than the $\gamma$-bounded SPAM noise model, where the post-measurement state changes by at most $\gamma$ in trace distance. Under our stronger model of corruption, we design an algorithm using non-adaptive measurements that can learn an unknown rank-$r$ state up to $\tilde{O}(\gamma\sqrt{r})$ in trace distance, provided that the number of copies is sufficiently large. We further prove an information-theoretic lower bound of $\Omega(\gamma\sqrt{r})$ for non-adaptive measurements, demonstrating the optimality of our algorithm. 
Our upper and lower bounds also hold for quantum state testing, where the goal is to test whether an unknown state is equal to a given state or far from it.

Our results are intriguingly optimistic and pessimistic at the same time. For general states, the error is dimension-dependent and $\gamma\sqrt{d}$ in the worst case, meaning that only corrupting a very small fraction ($1/\sqrt{d}$) of the outcomes could totally destroy any non-adaptive learning algorithm. However, for constant-rank states that are useful in many quantum algorithms, it is possible to achieve dimension-independent error, even in the worst-case adversarial setting.

\end{abstract}

\section{Introduction}


Learning properties of quantum states is a fundamental task in quantum computing and physics. In the canonical setting, we are given identical copies of an unknown state $\rho$ and need to design measurements and algorithms to learn information about $\rho$. Commonly studied problems include tomography~\cite{KRT14,ODonnellW17, HaahHJWY17,chen2023does,ADLY2025Paulinot}, where the goal is to learn the entire state description, and testing~\cite{BubeckC020,ChenLO22instance,liu2024role}, where the goal is to test whether $\rho$ satisfies some specific property. 

While extensive research has been done on these problems, a common assumption in many works (e.g., all the aforementioned works) is that the $\ns$ copies of the state are identical. This means that the copies are generated through identical and independent physical processes. They also assume that measurements are accurate and work exactly as described by their mathematical formulations. 
However, neither assumption may be true in practice, as the state preparation and measurement device may be prone to unknown error, especially in the current NISQ era~\cite{preskill2018quantum} of quantum computing.

To this end, many works have considered robustness under state preparation and measurement (SPAM) error \cite{brandao2020fast, yu2023robust,stilck2024efficient}. A general formulation considered in \cite{brandao2020fast} for single-copy measurements is that the state preparation and measurement lead to a total of $\gamma$ error in trace distance of post-measurement states. They show that $\Theta(\gamma)$ error is tight with sufficiently many samples.

We argue that SPAM noise may not be able to characterize all possible errors in the system. 
Unexpected physical factors, such as an unstable power source or even a subway train passing beneath the laboratory, may cause unwanted effects to the device. There may also be human factors. Some may be benign, like 
 inadvertent errors in measurement reading. Others may be malicious, like a hacker with sufficient quantum knowledge who tries to cause as much damage as possible. When any of these occur, the SPAM error assumptions may be violated, and the algorithms proposed by previous works may fail to work.

One may argue that improved quantum error correction may eventually reduce all physical errors below an acceptable threshold, especially given the recent progress~\cite{google2023suppressing,acharya2024quantum}. However, as discussed, many factors are non-physical and thus cannot be eliminated by error correction.
Therefore, it is necessary to propose a stronger corruption model for quantum state learning that takes into account these unexpected factors and design algorithms that are robust in this more challenging setting. The main question is as follows,
\begin{center}
    \fbox{Can we design quantum learning algorithms that still work even when things go against our favor?}
\end{center}

We focus on single-copy tomography and quantum state testing, and describe a corruption model suitable for these problems.

\subsection{Adversarial corruption model for single-copy tomography}
There are $\ns$ copies of $\rho$ and a random seed $R$, which can be viewed as an infinite random binary string. We can apply measurements $\POVM^\ns = (\POVM_1, \ldots, \POVM_\ns)$ to each copy, where $\POVM_i=\{M_x^i\}_{x\in\cX}$. For $i\ge 1$ let $\out_i$ be the outcome of measuring the $i$th copy with $\POVM_i$. All measurements can be chosen based on the random seed $R$. Define $\out^t=(\out_0,\out_1, \ldots, \out_t)$. 

\begin{description}
\item[Tomography.] The goal is to design a measurement scheme $\POVM^\ns$ and an estimator $\qest:\outset^\ns\mapsto \C^{\dims\times\dims}$ such that for all unknown state $\rho$,
\[
\probaOf{\tracenorm{\qest(\out^\ns)-\rho}\le \eps}\ge0.99.
\]
The probability is over the random seed $R$ and randomness in the measurement outcomes.

\end{description}

We also define a related problem of quantum state testing, which is useful to establish our lower bounds,
\begin{description}
    \item[Testing.]  Given a known target state $\qkn$, the goal is to test whether $\rho=\qkn$ or $\tracenorm{\rho-\qkn}>\eps$ with probability at least 0.8.
\end{description}

It is a standard fact that tomography is a harder problem than quantum state testing in terms of computation,
\begin{fact}
    Any tomography algorithm that learns an unknown state up to an accuracy of $\eps/2$ can be used to perform testing with an accuracy level of $\eps$.
    \label{fact:tomo-testing}
\end{fact}

When the state is $\rho$, the distribution of $\out_i,i\ge 1$ is determined by Born's rule,
\begin{equation}
    \p_\rho^{i}(x)=\Tr[M_x^i\rho],\label{equ:distr-i-cond}
\end{equation}
For $1\le t\le \ns$, we further define $\p_\rho^{\out^t}$ as the distribution of $\out^t$ when the state is $\rho$. For non-adaptive measurements, $\p_\rho^{\out^t}$ is a product distribution conditioned on the random seed $R$. The measurements $\POVM^\ns$ can also be chosen adaptively, i.e., $\POVM_t$ depends on previous outcomes $x^{t-1}$. In this case, $\p_\rho^{\out^t}$ would not be a product distribution in general.

In practice, we may have restrictions on the types of measurements that can be applied. We use $\povmset$ to denote the set of allowable measurements for each copy. 

\paragraph{Adversarial corruption.} We imagine that there is an adversary $\cA$ that can arbitrarily corrupt $\gamma$-fraction of the outcomes.
\begin{definition}[$\gamma$-adversarial corruption model]
    For non-adaptive measurement schemes, the interaction between the adversary and the algorithm is as follows,
\begin{enumerate}
    \item Measurements are applied to all copies of $\rho$ to obtain outcomes $(x_1, \ldots, x_\ns)$.
    \item The adversary $\cA$ arbitrarily changes a $\gamma$-fraction of them. It has perfect knowledge about the measurements used by the algorithm and can perform the corruption based on that knowledge. 
    \item The algorithm uses the corrupted outcomes $(y_1, \ldots, y_\ns)$ to learn about the state $\rho$.
\end{enumerate}
\end{definition}

The adversarial model is inspired by formulations in the robust statistical estimation literature. Many works \cite{CDG17,kothari2018robust,DiakonikolasKP20, BakshiDHKKK20robustSoS, ilias2024SoSsubG} consider the case when  $\gamma$-fraction of the samples can be arbitrarily changed. ~\cite{acharya21manipulation} studied distributed estimation where a central server tries to learn the distribution from information-constrained messages sent by a group of users. In this setting, the adversary has the power to directly change the messages sent by the users. Beyond statistical estimation, there is also a vast literature on adversarial robust streaming algorithms for many fundamental problems, such as frequency estimation~\cite{Ben-EliezerJWY22, HassidimKMMS22}, linear regression~\cite{CherapanamjeriS23}, and clustering~\cite {DBravermanHMSSZ21}. 

In stark contrast to the abundance of adversarially robust classical algorithms, to our knowledge, there have been very few works, at least in quantum state learning, that consider a similar adversarial setting. Thus, we believe our work is an important step in understanding the adversarial robustness of quantum algorithms.

\paragraph{Adversarial corruption is stronger than SPAM noise.} 
We first review the $\gamma$-bounded SPAM model in~\cite{brandao2020fast}. Let $\gamma_s$ be the error due to state preparation, where the actual state we measure is $\tilde{\rho}$ that satisfies $\tracenorm{\tilde{\rho}-\rho}\le \gamma_s$.

The measurement error is characterized by the maximum induced error in trace norm. For a measurement $\POVM$ applied to a state $\rho$, we abuse the notation a bit by denoting the post-measurement state as $\POVM(\rho)$. Then
\[
\sup_{\rho\text{ state}}\tracenorm{\POVM(\rho)-\POVM'(\rho)}\le \gamma_{m}.
\]
Here $\POVM'$ is the actual measurement that is applied. Note that $\tilde{\rho}$ and $\POVM'$ could depend on measurements and differ for different copies. Bounded SPAM noise requires that $\gamma_s+\gamma_m\le \gamma$.

In our formulation, it may seem at first glance that our model is incomparable or even weaker than SPAM noise, since we still assume all states and measurements are perfect. However, because bounded SPAM noise only changes each outcome distribution by at most $\gamma$ total variation distance, the adversary can nearly simulate a $\gamma$-bounded SPAM noise by only corrupting at most $2\gamma$-fraction of the outcomes. See \cref{sec:adv-strong-formal} for formal proof.

This is similar to how in robust statistics, the strong contamination model where $\gamma$-fraction of the samples are adversarially changed, is stronger than the case when each sample is replaced by a sample from another distribution that is $\gamma$-far in total variation distance~\cite[Section 1.2.3]{diakonikolas2023algorithmic}. 

\subsection{Results}
We now state our results for robust quantum state learning under adversarial corruption.

\begin{theorem}
    Let $\rho$ be an unknown state with rank $r$. Under $\gamma$-corruption, there exists an algorithm using non-adaptive measurements that achieves an error of $\tildeO{\gamma\sqrt{r}}$, provided that the number of copies $\ns$ is sufficiently large.
\end{theorem}

Thus, for constant rank states, we can achieve dimension-independent error that nearly matches the $\Omega(\gamma)$ lower bound. Intuitively, this shows that prior knowledge of the structure of the state helps to achieve better performance under adversarial corruption. 

However, it leaves the question of whether it is possible to achieve dimension-independent error for general states whose rank may be large. We answer this question negatively.

\begin{theorem}
    Under $\gamma$-corruption, any quantum state testing algorithm using non-adaptive measurements must incur an error of $\eps=\Omega(\gamma\sqrt{r})$. In particular when $r=\dims$, the lower bound is $\Omega(\gamma\sqrt{\dims})$.
\end{theorem}

Since tomography is harder than testing (\cref{fact:tomo-testing}), the above error upper and lower bounds hold for both tomography and testing up to a constant of 2. This shows that our tomography algorithm is nearly optimal, at least when $\ns$ is sufficiently large. 

In contrast to bounded SPAM noise~\cite{brandao2020fast}, a dimension-dependent error is inevitable for adversarial corruption and can be as large as $\gamma\sqrt{\dims}$ for full-rank states. Since trace distance is at most a constant, and $\dims$ is exponential in the number of qubits in $\rho$, the adversary only needs to change \emph{exponentially small} fraction of the outcomes to entirely ruin any learning algorithm using non-adaptive measurements. 

Our lower bound also deviates from many works on robust mean estimation (with direct access to all potentially corrupted samples), where dimension-independent error~\cite{CDG17, acharya21manipulation} can be achieved even in high dimensions. Thus, non-adaptive single-copy quantum state learning is more susceptible to adversarial attacks than classical distribution estimation.

\paragraph{Why is dimension-independent error inevitable?} The main reason is that single-copy measurements are inherently information-constrained: any measurement outcome we obtain only tells very limited information about the state. Thus, to fully take advantage of the limited information, the post-processing algorithm must be sensitive to small changes. High sensitivity generally implies less robustness. This is similar to distributed estimation under local differential privacy and communication constraints~\cite{acharya21manipulation}, where a dimension-dependent error is also inevitable.

\paragraph{Quantum state testing} For the simpler problem of quantum state testing, we also propose an algorithm with provable error guarantees for all parameter regimes. See \cref{sec:testing}.

\subsection{Related work}

\paragraph{Quantum state learning}
Full-state quantum state tomography has been extensively studied in various measurement settings, including fully entangled measurement~\cite{ODonnellW17,HaahHJWY17}, single-copy measurements, both adaptive~\cite{chen2023does} and non-adaptive~\cite{HaahHJWY17}, and Pauli measurements~\cite{Yu2020Pauli,de_Gois_2024,keenan2025randommatrixtheorypauli,ADLY2025Paulinot}. Fully entangled measurements require a large and coherent quantum device and thus are less practical than single-copy measurements.

There is also extensive work on quantum state testing/certification under similar measurement settings \cite{ODonnellW15,BadescuO019,BubeckC020,Chen0HL22,liu2024role,liu2024restricted}. 
\cite{brandao2020adversarial} studies adversarial hypothesis testing. Although bearing the same word, their adversary only has the power to choose states in pre-defined hypothesis classes that the algorithm needs to distinguish. Thus, it is a different problem formulation from ours.

Another related problem is shadow tomography~\cite{Aaronson20,huang2020predicting, chen2024pauli}, which aims to learn the expectation values of a finite set of observables of interest.

\paragraph{Quantum state learning with noise}
Many works have considered learning with noise in the system, particularly SPAM noise. The noise models in \cite{yu2023robust,stilck2024efficient} require problem and device-specific assumptions.  \cite{rambach2021robust} designs and implements an algorithm robust to specific types of statistical and environmental noise. The bounded SPAM model in~\cite{brandao2020fast} is the most general formulation among these works.

Error mitigation \cite{endo2018practical, cai2023quantum} is another practical approach to reduce the effect of noise. However, such an approach usually also requires prior assumptions on the noise model.
\cite{Jayakumar2024universalframework} studies simultaneous tomography of state and noise. However, due to the inherent ill-condition of the problem, the state and noise model can only be learned up to an ambiguity factor. To pinpoint the unknown state thus requires prior knowledge of the noise model.

We point out that these works mainly focus on physical errors and do not address other unexpected and even adversarial factors.

\paragraph{Agnostic learning and tolerant testing of quantum states} In the quantum agnostic learning framework, we consider a class of quantum states, $\mathcal{C}$. Given a target state $\rho$ and an error parameter $\epsilon$, the goal is to identify a state $\sigma \in \mathcal{C}$ such that
$$
\|\rho-\sigma\|_1 \leq \alpha\cdot \min_{\sigma^*\in \mathcal{C}} \|\rho-\sigma^*\|_1 + \epsilon
$$
for some constant $\alpha$. Chung and Lin~\cite{chung2018sample} and Badescu and O’Donnell~\cite{buadescu2021improved} showed that the sample complexity for agnostic learning is logarithmic in the size of the concept class $\mathcal{C}$. Moreover, existing works~\cite{bakshi2024learning, grewal2024agnostic, chen2024stabilizer} have studied agnostic learning for specific classes—such as stabilizer states and product states—and have demonstrated the existence of polynomial-time algorithms. A related setting is tolerant testing: given a property $\mathcal{P}$ of quantum states and two error parameters $\epsilon < \epsilon'$ (e.g., for product states), the task is to decide whether a quantum state $\rho$ is $\epsilon$-close to $\mathcal{P}$ or $\epsilon'$-far from it. This setting was recently examined for stabilizer states, yielding algorithms with running time $\mathsf{poly}(1/\epsilon', n)$~\cite{arunachalam2024note,arunachalam2024polynomial,iyer2024tolerant}.

The robust learning model considered in this work is fundamentally distinct from the agnostic learning and tolerant testing frameworks. In those settings, it is typically assumed that all copies of the state $\rho$ are affected by the same noise that causes a deviation from the ideal concept class. In contrast, our robust setting allows arbitrary, adversarial noise to be introduced in a $\gamma$-fraction of the input copies. In other words, an algorithm is deemed robust if it succeeds with high probability regardless of the noise pattern or the specific $\gamma$-fraction of corrupted input copies.

\paragraph{Classical distribution estimation}
Quantum state tomography and testing share many resemblance with classical distribution estimation. Single-copy measurement is similar to the distributed setting, where samples are distributed across multiple devices and limited information is sent to the central server~\cite{duchi2013local,barnes2019lower,ACLST22iiuic}. The works of \cite{liu2024role,liu2024restricted,ADLY2025Paulinot,ADLY2025PauliOpt} generalize the information-theoretic lower bound framework for distributed inference of discrete distributions \cite{AcharyaCT19,ACLST22iiuic} to single-copy quantum state tomography and testing, formally showing the deep connection between the two problems.

Robust statistics~\cite{huber1964robust,lugosi2021robust} is a long-standing  field of active research. There has been a large body of recent work that aims to design time-efficient algorithms for various statistical tasks such as mean estimation~\cite{CDG17,DiakonikolasKP20,7782980,Hopkins2020robustheavy,hopkins2025subGmean}, learning mixtures of Gaussians~\cite{BakshiDHKKK20robustSoS,Liu2023robustmixture}, and covariance estimation~\cite{kothari2018robust,BakshiDHKKK20robustSoS,ilias2024SoSsubG}, among others. As will be discussed in~\cref{sec:tomography-overview}, our tomography algorithm relies on robust covariance estimation.

\section{Technical overview}
\subsection{Lower bound }
Under our strong corruption model, the adversary has the power to arbitrarily change a small fraction of the outcomes to make them appear as if they come from some other states that are far away. The key to proving a lower bound is then to quantify how much damage could be caused given a measurement scheme. To this end, we need (1) an information-theoretical framework that takes adversarial corruption into account, and (2) a mathematical tool that quantifies the weakness of any measurement under adversarial attack.

\paragraph{Information-theoretic framework under adversarial corruption}\label{sec:lower-framework-overview}
The general recipe to prove lower bounds is to find a reference state $\rho$ and a distribution $\cD$ over states sufficiently far from $\rho$, say with trace distance at least $\eps$. Then we argue that for all measurements, the distribution of measurement outcomes would be close. By a standard hypothesis testing argument, e.g., Le Cam's method~\cite{LeCam73,yu1997assouad}, we can conclude that it is hard for any measurement to distinguish $\rho$ and states that are $\eps$-far. This establishes a \emph{testing} lower bound with an error of $\eps$, which also holds for tomography. 

Under our corruption model, since the adversary can arbitrarily change a fraction of the outcomes, it gains the power to make the outcome distributions even closer, and thus cause more difficulty for the algorithm to distinguish between different states. 
Then, we need to characterize how much closer the outcome distributions can be under adversarial corruption. This is achieved through the earth-mover distance, or a coupling argument.

To illustrate our argument, let us fix a set of measurements $\POVM_1, \ldots \POVM_{\ns}$\footnote{We assume that the adversary knows the shared random seed, so all measurements are fixed to the adversary}. We obtain outcomes $x^\ns=(x_1, \ldots, x_\ns)$ when measuring the quantum states. When the state is $\rho$, we denote the outcome distribution as $\p_{\rho}^{x^\ns}$. To successfully test between quantum states that are $\eps$ far, we need to solve the following hypothesis testing problem,
\[
H_0: \p_{\rho}^{x^\ns},\quad H_1 :\expectDistrOf{\sigma\sim \cD}{\p_{\sigma}^{x^\ns}}.
\]
How much fraction of the outcomes $x^\ns$ one needs to change to make one distribution $H_0$ look exactly like the other $H_1$? This can be characterized using coupling. A coupling $\Pi$ between $H_0$ and $H_1$ is a joint distribution $(X^\ns, Y^\ns)$ such that each marginal is distributed according to $H_0$ and $H_1$, respectively. For any coupling $\Pi$, we can define
\begin{equation}\label{equ:coupling-distance}
        \expectDistrOf{(X^\ns, Y^{\ns})\sim \Pi}{\ham{X^{\ns}}{Y^{\ns}}},
\end{equation}
which is the expected number of coordinates that $X^\ns, Y^\ns$ differ. We can view $\Pi$ as a randomized mapping. Then an adversary can use $\Pi$ as a strategy to change the outcomes.  If the above distance is at most $\gamma$,  then using $\Pi$, the adversary only changes $\gamma$-fraction of the outcomes \emph{in expectation} to make outcomes from $H_0$ look exactly as if they come from $H_1$.

One caveat is that we enforce a hard constraint on the number of outcomes that an adversary can change, so $\Pi$ cannot be directly used. However, we can make small changes to the adversarial strategy to meet the hard constraint while still ensure the corrupted outcomes are close in total variation distance.

Usually, to prove the tighter lower bound for tomography, we need to design a more difficult decision problem via more advanced arguments such as Fano's, Holevo~\cite{holevo1973statistical}, or Assouad~\cite{Assouad83,ADLY2025Paulinot}. However, since the power of the adversary is so strong, a testing lower bound suffices even for tomography. This can also be seen from bounded SPAM noise
~\cite{rambach2021robust} and robust statistical estimation literature \cite{diakonikolas2023algorithmic,acharya21manipulation}.

\paragraph{Measurement information channel quantifies the weakness under adversarial attack}
Having established a general framework, we just need to upper bound the minimum achievable expected coupling distance \eqref{equ:coupling-distance} for any measurement. The \emph{measurement information channel}~\cite{liu2024restricted} is helpful for our purpose.

\begin{definition}
    For a POVM $\POVM=\{M_x\}_{x\in \outset}$, the measurement information channel $\Luders_{\POVM}:\C^{\dims\times\dims}\mapsto\C^{\dims\times\dims}$ is defined as 
    \[
    \Luders_{\POVM}(\rho)=\sum_{x\in \outset}\frac{M_x}{\Tr[M_x]}\Tr[M_x\rho].
    \]
    It has a matrix form of 
    \[
    \Choi_{\POVM}=\sum_{x\in\outset}\frac{\vvec{M_x}\vadj{M_x}}{\Tr[M_x]},
    \]
    which satisfies $\Choi_{\POVM}\vvec{\rho}=\vvec{\Luders_{\POVM}(\rho)}$. We sometimes drop the subscript when the measurement is clear from context.
\end{definition}

We instantiate our lower bound framework with $\qmm=\eye_\dims/\dims$ as the reference state and a distribution $\cD$ parameterized by binary vectors in $z=(z_1, \ldots, z_\ell)\{-1, 1\}^\ell$,
\[
\sigma_z\eqdef\qmm+ \frac{\eps}{\sqrt{\dims\ell}}\sum_{i=1}^\ell z_iV_i, \quad \ell\sim \{-1,1\}^\ell.
\]
Here $V_i$'s are trace-0 orthonomal Hermitian matrices, $\Tr[V_i]=0$ and $\Tr[V_iV_j]=\indic{i=j}$. There are many ways we can choose the matrices: one option is the (normalized) Pauli observables. We choose $\ell=\dims^2/2$, which roughly matches the dimension of full-rank states. Using random matrix theory, we can argue that the above state is a valid quantum state with overwhelming probability (at least $1-\exp(-\dims)$).

For our choice of $\cD$, we can upper bound the minimum coupling distance in terms of the measurement information channel,
\[
\min_{\Pi}\expectDistrOf{(X^\ns, Y^{\ns})\sim \Pi}{\ham{X^{\ns}}{Y^{\ns}}}=\bigO{\frac{\ns\eps}{\dims}\sqrt{\sup_{\POVM}\tracenorm{\Luders_{\POVM}}}}.
\]
For all measurements, $\tracenorm{\Luders_{\POVM}}\le \dims$~\cite{liu2024restricted}. By our argument in \cref{sec:lower-framework-overview}, $\ns\gamma$ cannot be larger than the right-hand-side, else learning would not be possible. This proves that $\eps=\Omega(\gamma\sqrt{\dims})$. 

To obtain a lower bound that depends on the rank $r$, we simply restrict the construction to an $r$-dimensional subspace. Our lower bound techniques have regularity conditions that only apply to $r\ge 200$. However, since a lower bound of $\gamma$ hold for all states, we conclude that the error lower bound is $\Omega(\gamma\sqrt{r})$ for all rank-$r$ states.

\subsection{Tomography algorithm}\label{sec:tomography-overview}
\paragraph{Reduction to covariance estimation}
We first review the tomography algorithm in~\cite{guctua2020fast}, which reduces to estimating the (complex) covariance matrix of a vector-valued random variable. Thus to make it adversarially robust, it suffices to apply techniques in robust covariance estimation~\cite{kothari2018robust,diakonikolas2023algorithmic,BakshiDHKKK20robustSoS}.

We apply the uniform POVM to all copies, which consist of all suitably normalized rank-1 projectors, 
\[
\POVM_{unif}=\{\dims\qproj{v}\}_{v\in \Sp^{\dims} }.
\]
Each $\qbit{v}$ is drawn uniformly from the Haar measure, i.e., a uniform distribution over the complex unit sphere.
We can confirm using \eqref{equ:haar-k-moment} that $\int\dims\qproj{v}\dm v=\eye_\dims$ and thus it is a valid POVM. When applying the measurement to $\rho$, by Born's rule, we obtain a vector $\qbit{v}$ as an outcome which occurs with a probability density of 
\begin{equation}
\label{equ:uniform-povm-density}
    \dims\matdotprod{v}{\rho}{v}\dm v.
\end{equation}
We denote this distribution as $\cD(\rho)$. Again using \eqref{equ:haar-k-moment}, the expectation of $\qproj{v}$ is directly related to the unknown state $\rho$,
\begin{equation}
    \Sigma_\rho \eqdef \expectDistrOf{v\sim\cD(\rho)}{\qproj{v}} =\int\qproj{v}\matdotprod{v}{\rho}{v}\dm v= \frac{\eye_\dims+\rho}{\dims+1}.
    \label{equ:sigma-rho}
\end{equation}

Note that since $\expectDistrOf{v\sim \cD(\rho)}{\qbit{v}}=0$ by symmetry, $\expectDistrOf{v\sim\cD(\rho)}{\qproj{v}}$ is exactly the (complex) covariance matrix of the random variable $\qbit{v}$. The algorithm then estimates $\rho$ using the empirical covariance matrix,
\[
\hat{\rho}=(\dims+1)\frac{1}{\ns}\sum_{i=1}^\ns \qproj{v_i}-\eye_\dims.
\]

\paragraph{Robust covariance estimation in Hilbert-Schmidt/Frobenius distance}
To achieve the desired $\gamma\sqrt{r}$ error, we note that it suffices to obtain $\gamma$ error in terms of Hilbert-Schmidt distance. To see this, let $\rho$ be the unknown state and $\hat{\rho}$ be our estimate, and assume $\hsnorm{\rho-\hat{\rho}}\le \gamma$. By Cauchy-Schwarz,
\[
\tracenorm{\rho-\hat{\rho}}\le \sqrt{\dims}\hsnorm{\rho-\hat{\rho}}\le \gamma\sqrt{\dims}.
\]
Thus, we recover the worst-case error bound. To show a better bound when $\rho$ is low rank, it suffices to find the best rank-$r$ approximation of $\hat{\rho}$, call it $\tilde{\rho}$. Then $\tilde{\rho}-\rho$ has rank at most $2r$, and $\hsnorm{\rho-\tilde{\rho}}\le 2\gamma$ by triangle inequality. Rank-$2r$ Hermitian matrices only have at most $2r$ non-zero eigenvalues. Thus, applying Cauchy Schwarz, we obtain the $O(\gamma\sqrt{r})$ error as desired.

However, even though $\cD(\rho)$ is a sub-gaussian distribution, estimating the covariance with (relative) Frobenius distance is a challenging task. For general sub-gaussian distributions, dimension-dependent Frobenius error is unavoidable. Only a (relative) spectral error, i.e. operator norm distance, can be guaranteed ~\cite{kothari2018robust,ilias2024SoSsubG}, which is not sufficient for our purpose.

To our knowledge, only highly structured distributions can achieve dimension-independent Frobenius error. These include Gaussian \cite{diakonikolas2023algorithmic}, affine transform of uniform spherical distribution, distributions with sub-gaussian marginals \cite{BakshiDHKKK20robustSoS}, and Poincar\'e distributions\footnote{A distribution $\mathcal{D}$ over $\R^\dims$ is Poincar\'e with parameter $\sigma$ if for all differentiable functions $f:\R^\dims\mapsto \R$, $\variance{x\sim\mathcal{D}}{f(x)}\le \sigma^2\norm{\nabla f(x)}_2^2$.} \cite[Lemma 9.7]{BakshiKRTV24anti}.

Although the outcomes $\qbit{v}$ solely lies on the unit sphere, the distribution has a complicated structure. We start with the simple case where $\rho=\qproj{\psi}$ is a pure state. Then from \eqref{equ:uniform-povm-density}, the distribution has density of $0$ for vectors $\qbit{v}$ orthogonal to $\psi$, and highest density when $|\qdotprod{\psi}{v}|=1$. 
Moreover, we can show that $\expectDistrOf{\qproj{\psi}}{|\qdotprod{\psi}{v}|^2}\simeq 1/\dims^2$, so most of the probability mass should lie in some rings such that $|\qdotprod{\psi}{v}|\simeq 1/\dims$. When $\rho$ is a mixed state, then the distribution would be a mixture of these rings. 
It is unclear how to characterize this complicated distribution via affine transforms and/or projections of Gaussian or uniform spherical distributions. 

\paragraph{Hypercontractivity of outcomes from the uniform POVM}
The key to designing a robust tomography algorithm is to upper bound the high-order moments of $\qproj{v}$. We prove that the outcome distributions from the uniform POVM satisfies \emph{hypercontractivity},
\begin{definition}
     For all even integers $h\ge 2$, $\cD(\rho)$ satisfies $C$-hypercontractivity, i.e., there exists constant $C>0$ for all Hermitian matrix $M\in\C^\dims$, 
\[
\left(\expectDistrOf{v\sim\cD(\rho)}{\matdotprod{v}{M}{v}^{h}}\right)^2\le \frac{(Ch)^{2h}}{\dims^{2h}}(\Tr[M^2]+\Tr[M]^2)^h.
\]
Mathematically, one can view $\matdotprod{v}{M}{v}$ as a degree-2 polynomial in the real and imaginary parts of $\qbit{v}$ (which satisfies some algebraic structure defined by complex numbers). Physically, we can view the measurement outcome as a pure state $\qproj{v}$, and $\matdotprod{v}{M}{v}$ is the expectation value of an \emph{observable} $M$ (in quantum mechanics, anything that can be observed about the state must be represented as a Hermitian matrix). Regardless of how one chooses to interpret $\matdotprod{v}{M}{v}$, it should have well-bounded high-order moments. This is similar to hyper-contractivity proposed in~\cite{BakshiDHKKK20robustSoS}. We extend the notion to complex random vectors for complex covariance estimation.

To see how this property comes into play, we describe a highly inefficient algorithm and ``prove'' its guarantees using a slightly hand-wavy argument. 

We fix some $h=2t$ to be chosen later. Let $\qbit{u_1},\ldots,\qbit{u_\ns}$ be a set of clean samples. With a sufficiently large number of samples $\ns$, we can guarantee that with high probability, the empirical average is close to its true mean, so the finite samples also satisfy hyper-contractivity.
\[
\frac{1}{\ns}\sum_{i=1}^\ns\matdotprod{v_i}{M}{v_i}^h\le \frac{(Ch)^{h}}{\dims^{h}}(\Tr[M^2]+\Tr[M]^2)^{h/2}.
\]

With an even larger $\ns$ (perhaps exponentially large), we can further guarantee that hypercontractivity holds for any subset of the samples $S\subseteq[\ns]$ with $|S|=(1-\gamma)\ns$. 

Let $\qbit{v_1},\ldots,\qbit{v_\ns}$ be the $\gamma$-corrupted, so $\qbit{v_i}=\qbit{u_i}$ for at least $(1-\gamma)\ns$ of the samples. 
We can brute-force search for a subset $S$ that satisfies hyper-contractivity. This may not be feasible since we need to verify for all Hermitian $M$, but suppose we are able to find such a subset. Then we output our estimate $\hat{\Sigma}$ as the empirical covariance over the set $S$, $\hat{\Sigma}=\frac{1}{|S|}\sum_{i\in S}\qproj{v_i}$. 

Note that $\ns$ is very large, so we can assume that the true covariance is roughly the empirical average of clean samples, i.e. $\Sigma_\rho\simeq \frac{1}{|S|}\sum_{i\in S}\qproj{u_i}$ . Then for all Hermitian $M$,
\begin{align*}
    \Tr[M(\hat{\Sigma}-\Sigma_\rho)]&\simeq  \frac{1}{|S|}\sum_{i\in S}^\ns(\matdotprod{v_i}{M}{v_i}-\matdotprod{u_i}{M}{u_i})\\
    &=\frac{1}{|S|}\sum_{i\in S} \indic{\qbit{u_i}\ne\qbit{v_i}}(\matdotprod{v_i}{M}{v_i}-\matdotprod{u_i}{M}{u_i}).
\end{align*} 
Raising to the power of $h=2t$ and using Jensen's inequality,
\begin{align*}
    \Tr[M(\hat{\Sigma}-\Sigma_\rho)]^{2t}&\le 2^{2t-1}\Paren{\frac{1}{|S|}\sum_{i\in S} \indic{\qbit{u_i}\ne\qbit{v_i}}\matdotprod{v_i}{M}{v_i}}^{2t}+ 2^{2t-1}\Paren{\frac{1}{|S|}\sum_{i\in S} \indic{\qbit{u_i}\ne\qbit{v_i}}\matdotprod{u_i}{M}{u_i}}^{2t}
\end{align*}
Using H\"older's inequality that for any measure $\dm x$, $\int |fg| \dm x \le (\int |f|^p \dm x)^{1/p}(\int |g|^{q}\dm x)^{1/q}$ for $1/p+1/q=1$, and that $\indic{\qbit{u_i}\ne\qbit{v_i}}^2=\indic{\qbit{u_i}\ne\qbit{v_i}}$, we have
\begin{align*}
    \Paren{\frac{1}{|S|}\sum_{i\in S} \indic{\qbit{u_i}\ne\qbit{v_i}}\matdotprod{v_i}{M}{v_i}}^{2t}&\le  \Paren{\frac{1}{|S|}\sum_{i\in S}\indic{\qbit{u_i}\ne\qbit{v_i}}^{2t}}^{2t-1}\Paren{\frac{1}{|S|}\sum_{i\in S}\matdotprod{v_i}{M}{v_i}^{2t}}\\
    &\le \frac{(C'\gamma)^{2t-1}t^{2t}}{\dims^{2t}}(\hsnorm{M}^2+\Tr[M]^2)^t.
\end{align*}
Here we used that the potentially corrupted samples in $|S|$ satisfies hypercontractivity, and that  $\sum_{i\in S}\indic{\qbit{u_i}\ne v_i}\le \gamma\ns$ since the adversary can only change $\gamma$-fraction. We can bound the second term similarly. Therefore combining the two parts, we have for some constant $c$,
\[
\Tr[M(\hat{\Sigma}-\Sigma_\rho)]^{2t}\le\frac{(c\gamma t)^{2t}}{\gamma\dims ^{2t}} (\hsnorm{M}^2+\Tr[M]^2)^t.
\]
Since it holds \emph{for all} $M$, we can set $M=\hat{\Sigma}-\Sigma_\rho$. Note that $\Tr[M]=0$. Thus,
\[
\hsnorm{\hat{\Sigma}-\Sigma_\rho}^{4t}\le \frac{(c\gamma t)^{2t}}{\gamma\dims ^{2t}}\hsnorm{\hat{\Sigma}-\Sigma_\rho}^{2t}\implies \hsnorm{\hat{\Sigma}-\Sigma_\rho}\le \frac{c\gamma t}{\gamma^{1/(2t)}\dims}.
\]
Therefore, setting $t=O(\log(1/\gamma))$ yields an error of $\tildeO{\gamma /\dims}$. Setting $\hat{\rho}=(\dims +1)\hat{\Sigma}-\eye_\dims$, we conclude that the Frobenius error for estimating $\rho$ is $\tildeO{\gamma}$ as desired.

The major issue in our analysis is how to verify that hypercontractivity holds for all Hermitian matrices. To this end, we use the sum-of-squares technique~\cite{hopkins2017SoS,bakshi2020outlierrobust, BakshiDHKKK20robustSoS}. As long as we can show that hypercontractivity can be written as a degree $O(h)$ sum-of-squares polynomial in the matrix indeterminant $M$, then we can design a sum-of-squares algorithm to verify and ``find'' a subset $S$ that satisfies hypercontractivity. The full algorithm description and proof are given in \cref{sec:tomography-algorithm}.

\end{definition}

\section{Preliminaries}
\subsection{Quantum state and POVM}
We use the Dirac notation $\qbit{\psi}$ to denote a vector in $\C^{\dims}$. We use $\qbit{j}$ to denote the vector with 1 at the $j$th coordinate and 0 everywhere else. 

$\qadjoint{\psi}\eqdef(\qbit{\psi})^\dagger$ is a row vector. $\qdotprod{\psi}{\phi}$ is the inner product of $\qbit{\psi}$ and $\qbit{\phi}$. We denote the set of all $\dims\times\dims$ Hermitian matrices by $\Herm{\dims}$. A $\dims$-dimensional quantum system is described by a positive-semidefinite Hermitian matrix $\rho\in\Herm{\dims}$ with $\Tr[\rho]=1$. We assume $\dims=2^{\nqubits}$ where $\nqubits$ is the number of qubits in the system. 

A special case is the maximally mixed state $\qmm\eqdef \eye_\dims /\dims$. 

Measurements are formulated as \emph{positive operator-valued measure} (POVM). Let $\outset$ be an outcome set. Then a POVM $\POVM=\{M_x\}_{x\in \outset}$, where $M_x$ is p.s.d. and $\sum_{x\in \outset}M_x=\eye_\dims$. Let $X$ be the outcome of measuring $\rho$ with $\POVM$, then the probability observing $x\in\outset$ is given by the \emph{Born's rule},
\[
\probaOf{X=x}=\Tr[\rho M_x].
\]
This definition can be extended to an infinite outcome set.

\subsection{Complex matrices and operators}

\paragraph{Complex matrices}
Let $A,B\in\C^{\dims\times\dims}$. Inner product between $A$ and $B$ is defined as $\hdotprod{A}{B}\eqdef\Tr[A^\dagger B]$. For Hermitian matrices $A,B$, $\hdotprod{A}{B}=\hdotprod{B}{A}\in \R$. Thus the subspace of Hermitian matrices $\Herm{\dims}$ is a \textit{real} Hilbert space (i.e. the associated field is $\R$) with the same matrix inner product. 

Let $\{\qbit{j}\}_{j=0}^{\dims-1}$ be the natural basis for $\C^\dims$, we define vectorization as $\VecOp(\qoutprod{i}{j})\eqdef \qbit{j}\otimes \qbit{i}$. For convenience we denote $\vvec{A}\eqdef\VecOp(A)$. Matrix inner product satisfies $\hdotprod{A}{B}=\vvdotprod{A}{B}$. 

\paragraph{Linear superoperators} Let $\mathcal{N}:\C^{\dims\times \dims}\mapsto \C^{\dims\times \dims}$ be a linear operator over $\C^{\dims\times \dims}$. Every superoperator $\mathcal{N}$ has a matrix representation $\Choi(\mathcal{N})\in \C^{\dims^2\times\dims^2}$ that satisfies $\vvec{\mathcal{N}(X)}=\Choi(\mathcal{N})\vvec{X}$ for all matrices $X\in\C^{\dims\times\dims}$.

\paragraph{Schatten norms} Let $\Lambda=(\lambda_1, \ldots, \lambda_\dims)\ge 0$ be the \emph{singular values} of a linear operator $A$, which can be a matrix or a superoperator. {For Hermitian matrices, the singular values are the absolute values of the eigenvalues.} Then for $p\ge 1$, the \emph{Schatten $p$-norm} is defined as 
$
\|A\|_{S_p}\eqdef \|\Lambda\|_p
$. The Schatten norms of a superoperator $\mathcal{N}$ and its matrix representation $\Choi(\mathcal{N})$ match exactly, $\|\mathcal{N}\|_{S_p}=\|\Choi(\mathcal{N})\|_{S_p}$. Some important examples are trace norm $\tracenorm{A}\eqdef\|A\|_{S_1}$, Hilbert-Schmidt norm $\hsnorm{A}\eqdef\|A\|_{S_2}=\sqrt{\hdotprod{A}{A}}$, and operator norm $\opnorm{A}\eqdef\|A\|_{S_\infty}=\max_{i=1}^\dims\lambda_i$.

\subsection{Probability distances}
Let $\p$ and $\q$ be distributions over a finite domain $\mathcal{X}$. The \emph{total variation distance} is defined as 
\[
\totalvardist{\p}{\q}\eqdef\sup_{S\subseteq\mathcal{X}}(\p(S)-\q(S))=\frac{1}{2}\sum_{x\in\mathcal{X}}|\p(x)-\q(x)|.
\]
The KL-divergence  is
\[
\kldiv{\p}{\q}\eqdef\sum_{x\in\mathcal{X}}\p(x)\log\frac{\p(x)}{\q(x)}.
\]
The chi-square divergence
\[
\chisquare{\p}{\q}\eqdef \sum_{x\in\mathcal{X}}\frac{(\p(x)-\q(x))^2}{\q(x)}.
\]
By Pinsker's inequality and concavity of logarithm,
\[
2\totalvardist{\p}{\q}^2\le \kldiv{\p}{\q}\le \chisquare{\p}{\q}.
\]
We define $\ell_p$ distance as $
\norm{\p-\q}_p\eqdef\Paren{\sum_{x\in\mathcal{X}}{|\p(x)-\q(x)|^p}}^{1/p}.
$
All the definitions above can be extended to general probability measures.

\subsection{Haar measure}
The analysis of testing and tomography algorithms requires computing Haar integrals. Let $\Sp^{\dims}$ denote the complex unit sphere in $\C^{\dims}$. We use $\Haar{d}$ to denote the Haar measure over $\dims\times\dims$ unitary matrices. The Haar measure induces a unique unitarily invariant measure on  $\Sp^{\dims}$, which we denote as $\dm v$.

Detailed exposition of the representation theoretic techniques is beyond the scope of this paper. We only present the necessary definitions and results.

\paragraph{Permutation and cycles}
A permutation $\pi:[n]\mapsto[n]$ is a bijection over $[n]$. Let $\Sim_n$ be the set of all permutations over $[n]$. Every permutation can be decomposed into cycles. For example, the permutation $\pi=(1, 4, 2, 3)$ has cycles
\[
(1), (2,3,4).
\]
We use $\cycle(\pi)$ to denote the set of cycles in $\pi$. For each cycle $c\in \cycle(\pi)$, we use $|c|$ to denote the length of the cycle. In the example above $|(2,3,4)|=3$ and $|(1)|$=1.

\paragraph{Haar measure moments}
Let $\qbit{u}\in \C^{\dims}$ be a unit vector drawn uniformly from the unit sphere. The $\ab$-th order moment $\qproj{u}^{\otimes\ab}$ can be computed using Schur-Weyl duality, \cite[Eq.(14)]{guctua2020fast}.
\begin{equation}
    \expectDistrOf{\qbit{u}\sim\Haar{\dims}}{\qproj{u}^{\otimes\ab}}=\binom{\dims+k-1}{k}^{-1}P_{\text{Sym}^{(k)}}.
    \label{equ:haar-k-moment}
\end{equation}
Here $P_{\text{Sym}^{(k)}}=\frac{1}{k!}\sum_{\pi\in \Sim_{k}}P_{\pi}$ is the projection matrix onto the symmetric subspace of $(\C^{\dims})^{\otimes k}$, where $P_\pi$ is a permutation operator defined as
\[
P_\pi\qbit{\psi_1}\otimes\qbit{\psi_{\ab}}=\qbit{\psi_{\pi^{-1}(1)}}\otimes\cdots\otimes\qbit{\psi_{\pi^{-1}(\ab)}}.
\]
Permutation operators are unitary with $P_{\pi}^\dagger = P_{\pi}^{-1}=P_{\pi^{-1}}$.

Let $M_1, \ldots, M_k\in \C^{\dims\times\dims}$. For $\pi=(2,3,4,\ldots,k,1)$ (which circular shifts all numbers by one place), we have
\begin{equation}
    \Tr[P_\pi M_1\otimes\cdots\otimes M_k]=\Tr[M_1\cdots M_k].
    \label{equ:perm-trace}
\end{equation}

Using \eqref{equ:haar-k-moment}, we can compute the $k$-th moment of $\matdotprod{u}{M}{u}$ for any Hermitian matrix $M$,
\begin{align*}
    \expectDistrOf{\qbit{u}\sim\Haar{\dims}}{\matdotprod{u}{M}{u}^k}&=\expectDistrOf{\qbit{u}\sim\Haar{\dims}}{\Tr[\qproj{u}^{\otimes k}M^{\otimes k}]}\\
    &=\Tr[M^{\otimes k}\expectDistrOf{\qbit{u}\sim\Haar{\dims}}{\qproj{u}^{\otimes k}}]\\
    &=\binom{\dims+k-1}{k}^{-1}\Tr\left[P_{\text{Sym}^{(k)}}M^{\otimes k}\right].
\end{align*}
From \eqref{equ:perm-trace}, we have for each permutation $\pi\in \Sim_k$, 
\[
\Tr[P_\pi M^{\otimes k}]=\prod_{c\in \cycle(\pi)}\Tr[M^{|c|}].
\]
See also \cite[Equation (4)]{ChenLO22instance} for reference. Therefore,
\begin{equation}
    \expectDistrOf{\qbit{u}\sim\Haar{\dims}}{\matdotprod{u}{M}{u}^k}=\binom{\dims+k-1}{k}^{-1}\frac{1}{k!}\sum_{\pi\in\Sim_k}\prod_{c\in \cycle(\pi)}\Tr[M^{|c|}].
    \label{equ:haar-trace-moment}
\end{equation}

\subsection{Sum-of-squares and pseudo-distributions}
The sum-of-squares method is a powerful tool for robust statistical estimation tasks. We only describe the necessary preliminaries needed in our tomography algorithm and refer the readers to \cite{hopkins2017SoS, kothari2018robust, BakshiDHKKK20robustSoS} for more details.

\subsubsection{Pseudo-distributions}
Pseudo-distributions are a generalization of probability distributions. 
\begin{definition}
Let $\cX\subseteq \R^\dims$ be a finite set of vectors. A level-$\ell$ (or degree-$\ell$) pseudo-distribution over $\cX$ has a ``mass function'' $D(x)$ such that,
\[
\sum_{x\in \cX}D(x)=1,\quad \sum_{x\in \cX}D(x)f(x)^2\ge 0,
\]
for all polynomial $f$ such that $\deg(f)\le \ell/2$.
\end{definition}

For all function $f$ over $\R^\dims$, we define the pseudo-expectation as 
\[
\expectDistrOf{D}{f(x)}\eqdef\sum_{x}D(x)f(x).
\]

Given a set of polynomial constraints $\cA=\{f_1\ge0, \ldots f_m\ge 0\}$, we say that $D$ satisfies $\cA$ at degree $r$ if for every $S\subseteq[m]$ and every sum-of-squares polynomial $h$ with $\deg(h)=\sum_{i\in S}\max\{\deg(f_i),r\}$, 
\[
\expectDistrOf{D}{h(x)\prod_{i\in S}f_i(x)}\ge 0.
\]
We say that $D$ satisfies $\cA$ approximately at degree $r$ if the above inequalities are satisfied up to an error of $2^{-\dims \ell}\norm{h}_2\prod_{i\in S}\norm{f_i}_2$, where $\norm{\cdot}_2$ is the 2-norm of the coefficients of each monomial.

Due to the existence of an efficient separation oracle for moment tensors of pseudo-distributions \cite{nesterov2000squared,lasserre2001new}, we can efficiently find a level-$\ell$ pseudo-distribution that approximately satisfies a given set of constraints $\cA$ in time $(d+m)^{O(\ell)}$. See \cite[Fact 3.9]{BakshiDHKKK20robustSoS}.

Pseudo-expectations satisfy the Cauchy-Schwarz inequality,
\begin{fact}
    Let $f,g$ be polynomials of degree at most $d$ in $x\in \R^d$. Then for any degree-$d$ pseudo-distribution $D$, $\expectDistrOf{D}{fg}^2\le \expectDistrOf{D}{f^2}\expectDistrOf{D}{g^2}$.
\end{fact}

\subsubsection{Sum-of-squares proofs}
In this section, we define sum-of-squares proofs and introduce the necessary theoretical tools.
\begin{definition}[Sum-of-squares (SoS) proof]
    For polynomials $p,q$ in $x\in \R^\dims$, if there exists polynomials $s_1, \ldots, s_t$ of degree at most $k/2$ such that
    \[
    p-q=\sum_{i=1}^ts_i^2,
    \]
    we say that $p\ge q$ has a degree-$k$ sum-of-squares proof, denoted as $\sststile{k}{x}\{p\ge q\}$.
    
    Given axioms $\cA=\{f_i=0\}_{i\in [m]}\cup \{g_j\ge0\}_{j\in[n]}$ where $f_i,g_i$ are polynomials, we say that there exists a degree-$k$ proof of $p\ge q$ modulo $\cA$  if there exists polynomials $a_i,b_j,s_t$ such that $\deg(a_if_i)\le k$, $\deg(b_j^2g_j)\le k$, and $\deg(s_t^2)\le k$ for all $i,j,t$, and
    \[
    p-q=\sum_{t}s_t^2+\sum_{i=1}^m a_if_i+\sum_{j=1}^{n}b_j^2g_j,
    \]
    which we denote as $\cA\sststile{k}{x}\{p\ge q\}$.
\end{definition}

We now introduce some inference rules for SoS proofs.

\begin{fact}[SoS inference rules] SoS proofs satisfy the following rules.

\begin{description}
     \item[Addition] If $\cA\sststile{\ell}{x}\{f\ge 0, g\ge0\}$, then $\cA\sststile{\ell}{x}\{f+g\ge0\}$.

\item[Multiplication]  If $\cA\sststile{\ell}{x}\{f\ge 0\}$ and $\cA\sststile{\ell'}{x}\{g\ge 0\}$, then $\cA\sststile{\ell+\ell'}{x}\{fg\ge0\}$.

\item[Transitivity] If $\cA\sststile{\ell}{x}\cB$ and $\cB\sststile{\ell'}{x}\cC$, then $\cA\sststile{\ell\ell'}{x}\cC$.

\item[Substitution] Let $F:\R^{n}\mapsto \R^{m}, G:\R^{n}\mapsto \R^{k}, H:\R^p\mapsto\R^n$ be vector-valued polynomials. If $\{F\ge 0\}\sststile{\ell}{x}\{G\ge 0\}$, then $\{F(H)\ge 0\}\sststile{\ell\deg(H)}{x}\{G(H)\ge 0\}$.
\end{description}
   
\end{fact}

\begin{fact}
    Let $p$ be a univariate polynomial in $x\in \R$ and $p(x)\ge 0$ for all $x\in \R$. Then $\sststile{\deg(p)}{x}\{p(x)\ge 0\}$.
\end{fact}

\begin{fact}[Operator norm bound]
    Let $A$ be a $\dims\times\dims$ symmetric matrix and $v\in \R^\dims$. Then
    \[
    \sststile{2}{v}\{v^\top Av\le \opnorm{A}\norm{v}_2^2\}.
    \]
\end{fact}

\begin{fact}[SoS H\"older's inequality]
    Let $f_i, g_i$, $i\in [s]$ be indeterminates. Let $p$ be an even positive integer. Then
    \[
    \sststile{p^2}{f,g}\left\{\Paren{\frac1{s}\sum_{i=1}^sf_ig_i^{p-1}}^p\le \Paren{\frac{1}{s}\sum_{i=1}^sf_i^p}\Paren{\frac{1}{s}\sum_{i=1}^2g_i^p}^{p-1}\right\}.
    \]
    Setting $p=2$ yields the SoS Cauchy-Schwarz inequality.
\end{fact}

\begin{fact}[SoS almost triangle inequality]
    Let $f_1, \ldots, f_r$ be indeterminate. Then for all $t\ge 1$
    \[
    \sststile{2t}{f_1, \ldots, f_r}\left\{\Paren{\sum_{i=1}^rf_i}^{2t}\le r^{2t-1}\Paren{\sum_{i=1}^rf_i^{2t} } \right\}.
    \]
\end{fact}
\paragraph{Note on complex vectors and matrices}
In our problem, we mostly deal with complex vectors and Hermitian matrices. However, every complex vector $\qbit{x}\in \C^\dims$ can be written as $\qbit{x}=u+\img v$ and thus can be represented using $(u, v)\in \R^{2\dims}$. Therefore $\qdotprod{x}{x}=\norm{u}_2^2+\norm{v}_2^2$ is a real degree-2 polynomial in $u, v$.

Moreover, for all Hermitian matrix $A, B\in \C^{\dims\times\dims}$, the matrix inner product $\Tr[AB]\in \R$, so it can be viewed as a polynomial in the entries (both the real and imaginary part) of $A,B$. Thus, we can still apply the sum-of-squares technique for complex matrix expressions of this form.

\section{Lower bound}
\label{sec:lower-bound}
We describe the lower bound techniques in this section. The main technical result is a lower bound for testing maximally mixed states under $\gamma$-manipulation attack, expressed in terms of the \emph{measurement information channel}~\cite{liu2024restricted}.
\begin{theorem}
\label{thm:robust-lower}
    Let $\dims\ge 200$. With $\ns$ copies, using non-adaptive measurements, the error for testing maximally mixed state $\qmm$ under $\gamma$-manipulation with probability 0.8 must be at least
    \[
    \bigOmega{\frac{\gamma\dims}{\sqrt{\sup_{\POVM\in\povmset}\tracenorm{\Luders_{\POVM}}}}}.
    \]
\end{theorem}

\begin{corollary}
\label{cor:finit-outcome}
    Since $\tracenorm{\Luders_{\POVM}}\le \min\{\ab, \dims\}$ for $\POVM$ with at most $\ab$ outcomes, the additional cost is $\Omega(\gamma\dims/\sqrt{\min\{\ab, \dims\}})$ for these measurements.
\end{corollary}

\begin{corollary}
    For $\dims\ge r\ge 200$, let $\qmm^{(r)}$ be a maximally mixed state in an $r$-dimensional subspace of $\C^\dims$ (for example, if $\dims=2^{\nqubits}$ and $r=2^{N'}$ with $N'<N$, we can define $\qmm^{(r)}=\frac{1}{r}\sum_{i=0}^{r-1}\qproj{i}$.) Then \cref{cor:finit-outcome} holds with $\dims$ replaced by $r$. Thus, the lower bound to test rank-$r$ states is $\Omega(\gamma\sqrt{r})$.
\end{corollary}

\begin{remark}
    For simplicity, we present the definition of measurement information channel and our results for POVM with countable outcome set. However, our result and proof easily extends to continuous POVM by replacing the summation with integration over the outcomes.
\end{remark}

The remainder of this section is devoted to proving this lower bound. For simplicity of presentation, we first prove the lower bound for non-adaptive measurements.

We begin with a general framework that connects robust quantum state testing and the earth-mover distance (EMD).
Then, we describe the lower bound construction and apply the theoretical framework to prove our lower bounds.

\subsection{Robust testing lower bound via earth-mover distance} 
Suppose that the unknown state is $\rho$. The key insight is that if an adversary can change $\gamma$ proportion of the outcomes to make them appear as if they come from another state $\eps$-far from $\rho$, then we cannot hope to achieve better error than $\eps$.

This intuition is formalized in \cref{lem:emd}. We first define earth-mover distance with Hamming distance as the metric,
\begin{definition}
\label{def:emd}
    Let $\bfP_1,\bfP_2$ be two distributions over $\cX^{\ns}$ (note that they may not necessarily be product distributions). Let $\pi(\bfP_1, \bfP_2)$ be the set of couplings between the two distributions. The earth-mover distance is
    \[
    \EMD{\bfP_1}{\bfP_2}\eqdef \min_{\Pi\in \pi(\bfP_1,\bfP_2)}\expectDistrOf{(X^\ns, Y^{\ns})\sim \Pi}{\ham{X^{\ns}}{Y^{\ns}}}.
    \]
This is Kantorovich's formulation of optimal transport with Hamming distance as the cost. The existence of the minimum is due to the lower-semicontinuity of Hamming distance \cite[Chapter 6]{ambrosio2008gradient}.
\end{definition}

\begin{lemma}
\label{lem:emd}
    Let $\rho\in \C^{\dims\times\dims}$ be a fixed state and $\mathcal{D}$ be a distribution of states such that $\probaDistrOf{\sigma\sim\mathcal{D}}{\tracenorm{\rho-\sigma}>\eps}>0.9$. If for all measurements $\POVM^{\ns}=(\POVM_1, \POVM_2, \ldots, \POVM_\ns)$, 
    \[
    \EMD{\expectDistrOf{\sigma\sim\mathcal{D}}{\p_{\sigma}^{x^n}}}{\p_{\rho}^{x^n}}\le \frac{\gamma \ns}{2},
    \]
    then the error for quantum state testing must be at least $\eps$, which also implies the same lower bound for tomography.
\end{lemma} 

\begin{proof}
    Consider a game played by an adversary and a player with access to a testing algorithm. With equal probability (1/2), the adversary either (1) samples $\sigma\sim\mathcal{D}$ and sets all copies to be $\sigma$ or (2) sets the $\ns$ quantum states to $\rho$. In either case, the adversary can arbitrarily change $\gamma$-fraction of the outcomes. The player is asked to distinguish between the two cases using the potentially corrupted measurement outcomes. 
    
    Since $\probaDistrOf{\sigma\sim\mathcal{D}}{\tracenorm{\rho-\sigma}>\eps}>0.9$, if the testing algorithm can test whether $\sigma=\rho$ or $\tracenorm{\sigma-\rho}>\eps$ with probability at least 0.8, then the player should be able to guess correctly with probability at least
    \[
    0.8\cdot\frac{1}{2}(1+0.9)=0.76.
    \]

    We now design a strategy for the adversary to corrupt the measurement outcomes. By \cref{def:emd}, for all measurements $\POVM^\ns$, there exists a coupling $\Pi$ such that 
    \[
    \expectDistrOf{(X^\ns, Y^{\ns})\sim \Pi}{\ham{X^{\ns}}{Y^{\ns}}}\le \frac{\gamma\ns}{2},
    \]
    and the marginals are $X^{\ns}\sim \expectDistrOf{\sigma\sim\mathcal{D}}{\p_{\sigma}^{x^\ns}}, Y^{\ns}\sim\p_{\rho}^{x^n}$. By Markov's inequality,
    \[
    \probaDistrOf{(X^\ns, Y^{\ns})\sim\Pi}{\ham{X^{\ns}}{Y^{\ns}}>\gamma\ns}\le \frac12.
    \]
    
    We can view $Y^\ns$ as a randomized function of $X^\ns$, i.e. $Y^{\ns}=F(X^{\ns})$. This function can be used to alter the outcomes and fool the player. 
    
    In particular, the adversary does not make any changes in case (2) when the state is $\rho$.
    In case (1) when $\sigma$ is sampled from $\mathcal{D}$ and outcomes $X^\ns$ are obtained from the measurements,
    the adversary changes the outcomes to $F(X^{\ns})$ if $\ham{X^{\ns}}{F(x^{\ns})}\le \gamma\ns$, and keeps the original outcomes $X^\ns$ otherwise.  Equivalently, the adversary applies the following function to $X^\ns$,
    \[
    Adv(X^\ns)=\begin{cases}
        F(X^\ns), &\ham{X^{\ns}}{F(X^{\ns})}\le \gamma\ns,\\
        X^{\ns}, &o.w.
    \end{cases}
    \]

    The adversary changes at most $\gamma$-fraction of the outcomes in all scenarios.
    Let $Y^n$ be the outcomes in case (2). Since $Y^n$ has the same distribution as $F(Y^\ns)$, we can upper bound the total variation distance. 
    
    Denote $E$ as the event that $\ham{X^{\ns}}{F(X^{\ns})}\le \gamma\ns$, which happens with probability at least $1/2$. For any random variable $X$, we use the shorthand $X|E$ be the distribution of $X$ conditioned on $E$. Then,
    \begin{align}
        \totalvardist{Adv(X^\ns)}{Y^{\ns}}&= \totalvardist{Adv(X^\ns)}{F(X^{\ns})}\nonumber\\
        &=\totalvardist{Adv(X^\ns)|E}{F(X^{\ns})|E}\probaOf{E}+\totalvardist{{Adv(X^\ns)|E^c}}{{F(X^{\ns})|E^c}}\probaOf{E^c}\nonumber\\
        &=\totalvardist{{F(X^\ns)|E}}{{F(X^{\ns})|E}}\probaOf{E}+\totalvardist{{X^\ns|E^c}}{{F(X^{\ns})|E^c}}\probaOf{E^c}\nonumber\\
        &=\totalvardist{{X^\ns|E^c}}{{F(X^{\ns})|E^c}}\probaOf{E^c}\nonumber\\
        &\le \frac{1}{2}.\label{equ:dtv-adversary}
    \end{align}
    The final step is because $\probaOf{E^c}\le 1/2$ and $\totalvardist{\cdot}{\cdot}\le 1$. However, by standard argument for hypothesis testing (e.g.\cite{LeCam73, yu1997assouad}), the success probability is upper bounded by 
    \[
    \frac{1}{2}(1+\totalvardist{Adv(X^\ns)}{Y^{\ns}})\le  0.75.
    \]
    This leads to contradiction. Thus, the error for testing must be larger than $\eps$ if the testing algorithm can achieve success probability of at least 0.8.
\end{proof}

\subsection{Lower bound construction and EMD upper bound}
We now instantiate the framework with the following lower bound construction, originally introduced in \cite{liu2024role}. 
\begin{definition}
 \label{def:perturbation}
     Let $\dims^2/2\le\ell\le\dims^2-1$ and $\hbasis=(V_1, \ldots, V_{\dims^2}=\frac{\eye_\dims}{\sqrt{\dims}})$ be an orthonormal basis of $\Herm{\dims}$ (the space of $\dims\times\dims$ Hermitian matrices), and $\cd$ be a universal constant. Let  $\ptb=(\ptb_1, \ldots, \ptb_\ell)$ be uniformly drawn from $\{-1, 1\}^\ell$, we define $\sigma_{\ptb}=\qmm + \barDelta_{\ptb}$ where
     \begin{equation}
         \Delta_{\ptb} = \frac{\cd\eps}{\sqrt{\dims}}\cdot\frac{1}{\sqrt{\ell}}\sum_{i=1}^\ell \ptb_iV_i, \quad \barDelta_{\ptb}= \Delta_{\ptb}\min\left\{1, \frac{1}{2\dims \opnorm{\Delta_{\ptb}}}\right\},
         \label{equ:delta_z}
     \end{equation}
     Let the distribution of $\sigma_z$ be $\ptbDistr(\hbasis)$.
 \end{definition}
We apply \cref{lem:emd} with $\rho=\qmm$ and $\mathcal{D}=\ptbDistr(\hbasis)$. We can guarantee that $\sigma\sim \mathcal{D}$ is $\eps$ far from $\qmm$ with overwhelming probability. 

\begin{lemma}[{\cite[Corollary 4.4]{liu2024role}}]
\label{prop:perturbation-trace-distance}
    Let $\cd= 10\sqrt{2}$, $\ell\ge \dims^2/2$, $\eps<1/200$. Then for $\sigma\sim \ptbDistr(\hbasis)$,  $\|\sigma-\qmm\|_1\ge \eps$ with probability at least $1-2\exp(-\dims)$. 
\end{lemma}

We now upper bound the earth-mover distance for the distribution in \cref{def:perturbation}.
\begin{theorem}
\label{thm:emd-upper-general}
    Let $\sigma\sim \ptbDistr(\hbasis)$ as in \cref{def:perturbation}. For all measurements $\POVM^\ns\in \povmset^\ns$, 
\[
\EMD{\expectDistrOf{\sigma\sim\ptbDistr(\hbasis)}{\p_{\sigma}^{x^\ns}}}{\p_{\qmm}^{x^\ns}}\le 2\ns\eps\frac{\sqrt{\sup_{\POVM\in\povmset}\tracenorm{\Luders_{\POVM}}}}{\dims},
\]
where $\Luders_{\POVM}$ is the measurement information channel of $\POVM$.
\end{theorem}
\begin{proof}
    Since $\ptbDistr(\hbasis)$ is parameterized by binary vectors $z\in\{+1, -1\}^\ell$, we use the shorthand $\p_z^{x^\ns}$ to denote $\p_{\sigma_z}^{x^\ns}$. First, we note that
    \[
    \EMD{\expectDistrOf{z}{\p_z^{x^{\ns}}}}{\p_{\qmm}^{\ns}}\le \expectDistrOf{z}{\EMD{\p_z^{x^{\ns}}}{\p_{\qmm}^{\ns}}}.
    \]
    This is because the right-hand side is the expected Hamming distance of a coupling $(X^\ns, Y^\ns)$ between $\expectDistrOf{z}{\p_z^{x^{\ns}}}$ and $\p_{\qmm}^{\ns}$ which is constructed by the mixture of minimum couplings between $\p_z^{x^{\ns}}$ and $\p_{\qmm}^{x^\ns}$. Inequality holds by \cref{def:emd}.

    For each $z\in \{-1,1\}^\ell$, recall that $\p_z^{x^{\ns}}$ and $\p_{\qmm}^{x^\ns}$ are product distributions. Thus, we can construct a naive coupling between each outcome distribution $\p_z^{i}$  and $\p_{\qmm}^i$. We use the standard maximal coupling result,
    \begin{lemma}[{\cite[Theorem 2.12]{Hollander12}}]
        For any two distributions $p, q$, there exists a coupling $(X, Y)$ such that $\totalvardist{p}{q}=\probaOf{X\ne Y}$.
        \label{lem:maximal-coupling}
    \end{lemma}
    Therefore, constructing a maximal coupling between each  $\p_z^{i}$  and $\p_{\qmm}^i$ yields an expected Hamming distance of $\sum_{i=1}^\ns\totalvardist{\p_z^{i}}{\p_{\qmm}^i}$, which further upper bounds the earth-mover distance. Thus,
    \begin{align}
         \EMD{\expectDistrOf{z}{\p_z^{x^{\ns}}}}{\p_{\qmm}^{\ns}}&\le \sum_{i=1}^\ns\expectDistrOf{z}{\totalvardist{\p_z^{i}}{\p_{\qmm}^i}}\nonumber\\
         &\le \sum_{i=1}^{\ns}\expectDistrOf{z}{\sqrt{\frac{1}{2}\chisquare{\p_z^{i}}{\p_{\qmm}^{i}} }} &\text{Pinsker}\label{equ:lb-pinsker}\\
         &\le\sum_{i=1}^{\ns}\sqrt{\frac{1}{2}\expectDistrOf{z}{\chisquare{\p_z^{i}}{\p_{\qmm}^{i}} }} &\text{Concavity}\label{equ:lb-concave}
    \end{align}
    Thus it suffices to upper bound the expected chi-square divergence of each measurement outcome. 
    \begin{lemma}
        Let $\sigma_z$ be defined in \cref{def:perturbation}. For all measurement $\POVM$, 
        \[
        \expectDistrOf{z\sim\{-1, 1\}^\ell}{\chisquare{\p_z}{\p_{\qmm}}}\le\frac{\eps^2}{\ell}\sum_{i=1}^\ell\hdotprod{V_i}{\Luders_{\POVM}(V_i)}\le  \frac{\eps^2}{\ell}\tracenorm{\Luders_{\POVM}}.
        \]
        \label{lem:chi-square-upper}
    \end{lemma}
    \begin{proof}
    The proof follows by expanding the chi-square divergence using Born's rule. Let $\POVM=\{M_x\}_{x\in \cX}$, so $\p_z(x)=\Tr[M_x\sigma_z]$ and $\p_{\qmm}(x)=\Tr[M_x\qmm]$.
        \begin{align*}
            \expectDistrOf{z\sim\{-1, 1\}^\ell}{\chisquare{\p_z}{\p_{\qmm}}}&=\expectDistrOf{z}{\sum_{x}\frac{(\p_z(x)-\p_{\qmm}(x))^2}{\p_{\qmm}(x)}}\\
            &=\dims\cdot \expectDistrOf{z}{\sum_x\frac{\Tr[M_x\barDelta_z]^2}{\Tr[M_x]}}\\
            &\le \dims \cdot \expectDistrOf{z}{\sum_x\frac{\Tr[M_x\Delta_z]^2}{\Tr[M_x]}}
        \end{align*}
        The second step follows by linearity of trace and \cref{def:perturbation}. The inequality is because $\barDelta_z=a_z\Delta_z$ where $a_z\le 1$. The final expression can be further expressed in terms of the measurement information channel $\Luders$. Using the expression for $\Delta_z$ in \cref{def:perturbation},
        \begin{align*}
            \dims\cdot\expectDistrOf{z}{\sum_x\frac{\Tr[M_x\Delta_z]^2}{\Tr[M_x]}}&=\dims \cdot \frac{\eps^2}{\dims\ell}\cdot \expectDistrOf{z}{\sum_x\frac{\Tr[M_x\sum_{i=1}^\ell z_iV_i]^2}{\Tr[M_x]}}\\
            &=\frac{\eps^2}{\ell}\expectDistrOf{z}{\sum_x\sum_{i,j=1}^\ell z_iz_j\frac{\Tr[M_x V_i]\Tr[M_xV_j]}{\Tr[M_x]}}\\
            &=\frac{\eps^2}{\ell}\sum_{i=1}^\ell\sum_x\frac{\Tr[M_xV_i]\Tr[M_xV_i]}{\Tr[M_x]}\\
            &=\frac{\eps^2}{\ell} \sum_{i=1}^\ell \Tr\left[V_i\sum_{x}\frac{M_x\Tr[M_xV_i]}{
            \Tr[M_x]
            }\right]\\
            & = \frac{\eps^2}{\ell}\sum_{i=1}^\ell \Tr\left[V_i\Luders(V_j) \right]\\
            & =\frac{\eps^2}{\ell}\sum_{i=1}^\ell \vadj{V_i}\Choi\vvec{V_i}.
        \end{align*}
        This proves the first inequality in the statement. The final inequality follows by noting that $V_i$'s are orthonormal in the space of $\dims\times\dims$ matrices, and thus
        \[
        \tracenorm{\Luders}=\Tr[\Choi]=\sum_{i=1}^{\dims^2}\vadj{V_i}\Choi\vvec{V_i}\ge \sum_{i=1}^\ell \vadj{V_i}\Choi\vvec{V_i}.
        \]
        The inequality is because $\Choi$ is p.s.d. and thus $\vadj{V_i}\Choi\vvec{V_i}\ge 0$.
    \end{proof}
    Combining all parts and noting that $\ell=\dims^2/2$ proves the theorem.
\end{proof}
Our main lower bound result \cref{thm:robust-lower} is an immediate corollary of by setting $\gamma = 4\eps\frac{\sqrt{\sup_{\POVM\in\povmset}\tracenorm{\Luders_{\POVM}}}}{\dims}$ and solving for $\eps$.

\subsection{Adversarial corruption is stronger than SPAM: formal proof}
\label{sec:adv-strong-formal}
To conclude the section, we give a formal proof that our adversarial corruption model is at least as strong as $\gamma$-bounded SPAM noise up to some constants.

\begin{theorem}
    Given $\ns$ copies of an unknown state $\rho$ and measurements $\POVM_1, \ldots, \POVM_\ns$, let $X^\ns=(X_1, \ldots, X_n)$ be the random variable denoting the clean outcomes. Let $Y^\ns= (Y_1, \ldots, Y_n)$ be the random variable that denotes the outcomes due to some $\gamma$-bounded SPAM error. 
    
    Then, there exists an adversary $Adv$ that corrupts at most $2\gamma$-fraction of $x^{\ns}$ and guarantees 
    \[
    \totalvardist{Adv(X^\ns)}{Y^\ns}\le \exp(-2\gamma^2\ns).
    \]
\end{theorem}
\begin{proof}
    Since the SPAM noise is $\gamma$-bounded, we have for each outcome $\totalvardist{X_i}{Y_i}\le \gamma$ because trace distance is the maximum total variation discrepancy in the outcome distribution that can be caused for any measurement.
    
    Thus, upon receiving $X_i$, the adversary could apply maximal coupling~\cref{lem:maximal-coupling} to change $X_i$ to $X_i'$, and $X_i'$ would have the same distribution as $Y_i$, and 
    \[
    \probaOf{X_i\ne X_i'}=\totalvardist{X_i}{Y_i}\le \gamma.
    \]
    The adversary stops making changes at time $t$ when $\sum_{i\le t}\indic{X_i\ne X_i'}>2\gamma\ns$. Using Hoeffding's inequality,
    \[
    \probaOf{\sum_{i=1}^{\ns}\indic{X_i\ne X_i'}>2\gamma\ns}\le \exp(-2\gamma^2\ns).
    \]
    Using the same argument as \cref{lem:emd}, we have $\totalvardist{Adv(X_i)}{Y_i}\le \exp(-2\gamma^2\ns)$ as desired.
\end{proof}

Thus, for $\ns\ge 10/\gamma^2$, which is necessary to achieve an error of $O(\gamma)$ for any algorithm, one cannot tell whether the outcomes are results of SPAM noise or our adversary $Adv$ with probability larger than $\exp(-20)$. Thus, if two states cannot be distinguished with probability at least $0.8001$ in $\gamma$-bounded SPAM noise, they also cannot be distinguished with probability at least $0.8$ under $2\gamma$-adversarial corruption. This shows that our corruption model is stronger than SPAM up to constants.
\section{Robust tomography algorithm}\label{sec:tomography-algorithm}
We restate the algorithm without robustness guarantees.
\begin{algorithm}
\caption{Tomography algorithm using uniform POVM}
\label{alg:tomograph-non-robust}
Measure all copies with $\POVM_{unif}$. Obtain $\qbit{v_1}, \ldots \qbit{v_\ns}$ as outcomes.

Return $\hat{\rho}=(\dims+1)\frac{1}{\ns}\sum_{i=1}^{\ns}\qproj{v_i}-\eye_\dims$.
\end{algorithm}

The benefit of \cref{alg:tomograph-non-robust} is that the proof only requires the second moment of the Haar measure, thus, it can be efficiently implemented using 2-designs. We will first prove its error guarantee under adversarial attack. We then introduce another algorithm that achieves the optimal $\gamma\sqrt{\dims}$ error rate, but requires much higher moments of the Haar measure.

\begin{theorem}
    Let $\rho$ be the unknown state with rank $\rk(\rho)=r$. With $\gamma$-corrupted measurement outcomes from $\ns$ copies, \cref{alg:tomograph-non-robust} achieves an error of
    \[
    \tracenorm{\hat{\rho}-\rho}\le \tildeO{\frac{r\sqrt{\dims}}{\sqrt{\ns}}}+2\gamma(\dims+1).
    \]
\end{theorem}
\begin{proof}
    Let $\qbit{v_1}, \ldots, \qbit{v_\ns}$ be uncorrupted outcomes. Let $S\subseteq[\ns]$ be the indices chosen by the adversary and $\qbit{\tilde{v_i}}$ be the corrupted outcomes for $i\in S$. Then \cref{alg:tomograph-non-robust} outputs an estimate of
    \[
    \hat{\rho}=(\dims+1)\frac{1}{\ns}\Paren{\sum_{i\in S}\qproj{\tilde{v}_i}+\sum_{i\notin S}\qproj{v_i}}-\eye_\dims.
    \]
    Let $\tilde{\rho}$ be the estimate we would have obtained had the samples not been corrupted,
    \[
    \tilde{\rho}=(\dims+1)\frac{1}{\ns}\sum_{i=1}^{\ns}\qproj{v_i}-\eye_\dims.
    \]
    Then, by triangle inequality
    \begin{align*}
        \tracenorm{\hat{\rho}-\rho}&\le \tracenorm{\hat{\rho}-\tilde{\rho}}+\tracenorm{\tilde{\rho}-\rho}\\
        &=\frac{\dims+1}{\ns}\tracenorm{\sum_{i\in S}(\qproj{v_i}-\qproj{\tilde{v_i}})}+\tracenorm{\tilde{\rho}-\rho}\\
        &\le \frac{\dims+1}{\ns}\sum_{i\in S}\tracenorm{\qproj{v_i}-\qproj{\tilde{v_i}}}+\tracenorm{\tilde{\rho}-\rho}
    \end{align*}
    By the error bounds for tomography with uncorrupted outcomes \cite[Theorem 1]{guctua2020fast}, $\tracenorm{\tilde{\rho}-\rho}=\tilde{O}(r\sqrt{\dims}/\sqrt{\ns})$ with probability at least 0.99. Furthurmore, for each $i\in S$, $\tracenorm{\qproj{v_i}-\qproj{\tilde{v}_i}}\le 2$. Since $|S|\le \gamma\ns$, combining all parts proves the theorem.
\end{proof}


\subsection{Robust tomography algorithm with near-optimal error}
To robustly estimate the state, we only need to apply robust covariance estimation for the complex random variable $\qbit{v}$. To achieve the optimal error rate in \cref{thm:robust-lower}, it suffices to show that $\hsnorm
{\hat{\rho}-\rho}=O(\gamma)$. 

\begin{theorem}
    Let $\gamma<1/10$ and $\ns=\dims^{O(\log 1/\gamma)}/\gamma^2$. There exists a non-adaptive robust quantum state tomography algorithm that, given $\gamma$-fraction corrupted measurement outcomes, obtains an estimate $\hat{\rho}$ of the unknown state $\rho$ with $\hsnorm{\rho-\hat{\rho}}\le \tildeO{\gamma}$.
    \label{thm:tomography-alg}
\end{theorem}
Then, via a projection argument and Cauchy-Schwarz, we can achieve an accuracy of $O(\gamma\sqrt{\rk})$ in trace distance.
\begin{corollary}
    Let $\rho$ be a rank-$r$ state and $\hat{\rho}$ be the estimate obtained by \cref{thm:tomography-alg}. Define $\tilde{\rho}=\arg\min_{\sigma:\rk(\sigma)\le r}\hsnorm{\sigma-\hat{\rho}}$ to be the closest rank-$r$ state to $\hat{\rho}$, which can be efficiently computed by the Eckart-Young-Mirsky Theorem. Then,
    \[
    \tracenorm{\tilde{\rho}-\rho}=\tilde{O}(\gamma\sqrt{r}).
    \]
\end{corollary}
\begin{proof}
    By triangle inequality and definition of $\tilde{\rho}$, we have
    \[
    \hsnorm{\tilde{\rho}-\rho}\le \hsnorm{\hat{\rho}-\rho}+\hsnorm{\tilde{\rho}-\hat{\rho}}\le 2\hsnorm{\hat{\rho}-\rho}
    \]
    Since both $\tilde{\rho}$ and ${\rho}$ have rank at most $r$, $\tilde{\rho}-\rho$ have rank at most $2r$. Let $\lambda_1, \ldots, \lambda_{2r}$ be its non-zero eigenvalues. Then by Cauchy-Schwarz,
    \[
    \tracenorm{\tilde{\rho}-\rho}\le \sqrt{2r}\hsnorm{\tilde{\rho}-\rho}\le 2\sqrt{2r}\hsnorm{\hat{\rho}-\rho}.
    \]
    Applying \cref{thm:tomography-alg} completes the proof.
\end{proof}

\textbf{Algorithm}
The algorithm relies on the sum-of-squares algorithm. Let $\cV =\{\qbit{v_1}, \ldots, \qbit{v_\ns}\}$ be $\gamma$-corrupted samples from $\cD(\rho)$.
We imagine that they are generated from clean samples $\{\qbit{u_1}, \ldots, \qbit{u_\ns}\}$.

We define the following set of constraints on variables $w_1, \ldots, w_n\in \R$ vectors $\qbit{x_1}, \ldots, \qbit{x_n}\in \C^{\dims}$, and Hermitian matrix $Q\in \Herm{\dims}$.

\begin{align}
   \constr_{\cV,\gamma}= \begin{cases}
       \sum_{i=1}^\ns w_i=(1-\gamma)\ns&\\
       w_i^2=w_i, & \forall i\in [\ns]\\
       w_i\cdot (\qbit{v_i}-\qbit{x_i})=0 &\forall i\in [\ns]\\
       \frac{1}{\ns}\sum_{i=1}^\ns \qproj{x_i} =\Sigma\\
       \frac{1}{\ns}\sum_{i=1}^\ns\matdotprod{x_i}{Q}{x_i}^{2t}\le (Ct)^{2t}\Paren{\frac{1}{\ns}\sum_{i=1}^\ns\matdotprod{x_i}{Q}{x_i}^2}^t &\text{Hypercontractivity} \\
       \frac{1}{\ns}\sum_{i=1}^\ns\matdotprod{x_i}{Q}{x_i}^2\le \frac{C}{(\dims+1)^2}(\hsnorm{Q}^2+\Tr[Q]^2). &\text{Bounded second moment}
   \end{cases}
   \label{equ:constraints}
\end{align}
The idea is to select a subset from $\cV$ that satisfies the hyper-contractivity and second moment constraints. 

Note that although written in the form of complex matrices and vectors, the above constraints are all \emph{real} polynomials in terms of the real and imaginary parts of $Q,\qbit{x_i},w_i$. The sum-of-squares algorithm then works as follows,
\begin{algorithm}
\caption{Robust tomography algorithm with optimal error}
    \begin{algorithmic}
        \State \textbf{Input:} $\cV =\{\qbit{v_1}, \ldots, \qbit{v_\ns}\}$ $\gamma$-corrupted samples from $\cD(\rho)$. 
        \State Let $t=\log(1/\gamma)$. Find a level-$O(t)$ pseudo-distribution $\tilde{\zeta}$ that satisfies the constraints $\constr_{\cV,\gamma}$
        \State Output $\hat{\rho}=(\dims+1)\expectDistrOf{\tilde{\zeta}}{\Sigma}-\eye_\dims$.
    \end{algorithmic}
\end{algorithm}

To prove the theorem, first we show that the constraint set $\constr_{\cV,\gamma}$ gives a degree $O(t)$ SoS upper bound of $\Tr[(\Sigma-\Sigma_\rho)Q]$ for all Hermitian $Q$.
\begin{lemma}
\label{lem:sos-constraint-bound}
    Let $\eta = C\sqrt{C}t(2\gamma)^{1-\frac{1}{2t}}$, and $\hat{\Sigma}=\frac{1}{\ns}\sum_{i=1}^\ns\qproj{u_i}$ be the empirical covariance of the clean samples. Then
    \[
    \constr_{\cV,\gamma}\sststile{4t}{\Sigma, Q}\left\{\Tr[(\Sigma-\hat{\Sigma})Q]^{2t}\le \frac{\eta^{2t}}{(\dims+1)^{2t}}(\hsnorm{Q}^2+\Tr[Q]^2)^t \right\}.
    \]
\end{lemma}
\begin{proof}
We first evaluate the expression $\Tr[(\Sigma-\hat{\Sigma})Q]$.
\begin{align*}
    \Tr[(\Sigma-\hat{\Sigma})Q]&=\Tr[\Sigma Q]-\Tr[\hat{\Sigma})Q]\\
    &=\frac{1}{\ns}\sum_{i=1}^\ns\matdotprod{x_i}{Q}{x_i}-\frac{1}{\ns}\sum_{i=1}^\ns\matdotprod{u_i}{Q}{u_i}\\
    &=\frac{1}{\ns}\sum_{i=1}^\ns(1-w_i\indic{\qbit{u_i}=\qbit{v_i}})(\matdotprod{x_i}{Q}{x_i}-\matdotprod{u_i}{Q}{u_i})+\frac{1}{\ns}\sum_{i=1}^\ns w_i\indic{\qbit{u_i}=\qbit{v_i}}(\matdotprod{x_i}{Q}{x_i}-\matdotprod{u_i}{Q}{u_i})\\
    &=\frac{1}{\ns}\sum_{i=1}^\ns(1-w_i\indic{\qbit{u_i}=\qbit{v_i}})(\matdotprod{x_i}{Q}{x_i}-\matdotprod{u_i}{Q}{u_i}).
\end{align*}
The final step is because the 
constraints $\constr_{\cV,\gamma}$ imply $w_i=\indic{\qbit{v_i}=\qbit{x_i}}$, and thus the second summation is 0. We proceed to use the SoS almost triangle inequality.
\begin{align}
    \constr_{\cV,\gamma}\sststile{4t}{\Sigma,Q}\Bigg\{\Tr[(\Sigma-\hat{\Sigma})Q]^{2t}&\le 2^{2t} \Paren{\frac{1}{\ns}\sum_{i=1}^\ns(1-w_i\indic{\qbit{u_i}=\qbit{v_i}})\matdotprod{x_i}{Q}{x_i}}^{2t} \nonumber\\
    &+ 2^{2t}\Paren{\frac{1}{\ns}\sum_{i=1}^\ns(1-w_i\indic{\qbit{u_i}=\qbit{v_i}})\matdotprod{u_i}{Q}{u_i} } ^{2t}\Bigg\}.\label{equ:constraint-proof}
\end{align}
We bound each term separately. Using SoS H\"older's inequality,
\begin{align*}
    \constr_{\cV,\gamma}\sststile{4t}{\Sigma,Q}\Bigg\{\Paren{\frac{1}{\ns}\sum_{i=1}^\ns(1-w_i\indic{\qbit{u_i}=\qbit{v_i}})\matdotprod{x_i}{Q}{x_i}}^{2t}&\le \Paren{\frac{1}{\ns}\sum_{i=1}^\ns(1-w_i\indic{\qbit{u_i}=\qbit{v_i}})^2}^{2t-1}\Paren{\frac{1}{\ns}\sum_{i=1}^\ns\matdotprod{x_i}{Q}{x_i}^{2t}}\\
    &\le(2\gamma)^{2t-1}(Ct)^{2t}\Paren{\frac{1}{\ns}\sum_{i=1}^\ns\matdotprod{x_i}{Q}{x_i}^2}^t\\
    &\le \frac{(2\gamma)^{2t-1}(Ct)^{2t}C^t}{(\dims+1)^{2t}}(\hsnorm{Q}^2+\Tr[Q]^2)^t\Bigg\}.
\end{align*}
The second inequality uses $(1-w_i\indic{\qbit{u_i}=\qbit{v_i}})^2=(1-w_i\indic{\qbit{u_i}=\qbit{v_i}})$ and $\sum_{i\in[\ns]}(1-w_i\indic{\qbit{u_i}=\qbit{v_i}})\le 2\gamma$, along with the hypercontractivity constraint in $\constr_{\cV, \gamma}$. The final step uses the bounded second moment constraint. 

Similarly, we can bound the second term in \eqref{equ:constraint-proof}. We will use the hypercontractivity of uncorrupted samples of the outcome distribution \cref{thm:hypercontracivet-sample}.
\begin{align*}
    \constr_{\cV,\gamma}\sststile{4t}{\Sigma,Q}\Bigg\{\Paren{\frac{1}{\ns}\sum_{i=1}^\ns(1-w_i\indic{\qbit{u_i}=\qbit{v_i}})\matdotprod{u_i}{Q}{u_i}}^{2t}&\le \Paren{\frac{1}{\ns}\sum_{i=1}^\ns(1-w_i\indic{\qbit{u_i}=\qbit{v_i}})^2}^{2t-1}\Paren{\frac{1}{\ns}\sum_{i=1}^\ns\matdotprod{u_i}{Q}{u_i}^{2t}}\\
    &\le \frac{(2\gamma)^{2t-1}(Ct)^{2t}C^t}{(\dims+1)^{2t}}(\hsnorm{Q}^2+\Tr[Q]^2)^t\Bigg\}.
\end{align*}
Combining with \eqref{equ:constraint-proof} proves the lemma.
\end{proof}

From here we can prove \cref{thm:tomography-alg}. Recall from the algorithm that $\tilde{\zeta}$ is the $O(t)$-level pseudo-distribution that satisfies the constraints $\constr_{\cV, \gamma}$. Let $\tilde{\Sigma}=\expectDistrOf{\tilde{\zeta}}{\Sigma}$ where $\Sigma$ is defined in $\constr_{\cV, \gamma}$. From \cref{lem:sos-constraint-bound}, we have that  for any Hermitian $Q\in \Herm{\dims}$,

\[
\expectDistrOf{\tilde{\zeta}}{\Tr[(\Sigma-\hat{\Sigma})Q]^{2t}}\le \Paren{\frac{\eta}{\dims+1}}^{2t}(\hsnorm{Q}^2+\Tr[Q]^2)^t,
\]
using that expectation of squares over pseudo-distributions preserves non-negativity. Then, by Cauchy-Schwarz for pseudo-distribtuions
\begin{align*}
    \Tr[(\expectDistrOf{\tilde{\zeta}}{\Sigma}-\hat{\Sigma})Q]^2&={\expectDistrOf{\tilde{\zeta}}{\Tr[(\Sigma-\hat{\Sigma})Q]}}^2\le \expectDistrOf{\tilde{\zeta}}{\Tr[(\Sigma-\hat{\Sigma})Q]^2}.
\end{align*}
Repeatedly applying Cauchy-Schwarz, and without loss of generality assume $t$ is a power of 2, we have
\[
{\expectDistrOf{\tilde{\zeta}}{\Tr[(\Sigma-\hat{\Sigma})Q]}}^{2t}\le \expectDistrOf{\tilde{\zeta}}{\Tr[(\Sigma-\hat{\Sigma})Q]^{2t}}
\]
Thus, 
\[
 \Tr[(\expectDistrOf{\tilde{\zeta}}{\Sigma}-\hat{\Sigma})Q]^2\le \frac{\eta^2}{(\dims+1)^2}(\hsnorm{Q}^2+\Tr[Q]^2).
\]
We then set $Q=\expectDistrOf{\tilde{\zeta}}{\Sigma}-\hat{\Sigma}$. Note that $\Tr[Q]=\expectDistrOf{\tilde{\zeta}}{\Tr[\Sigma]}-\Tr[\hat{\Sigma}]=0$. We thus have 
\[
\hsnorm{\expectDistrOf{\tilde{\zeta}}{\Sigma}-\hat{\Sigma}}^4\le \Paren{\frac{\eta\hsnorm{\expectDistrOf{\tilde{\zeta}}{\Sigma}-\hat{\Sigma}}}{\dims+1}}^2\implies \hsnorm{\expectDistrOf{\tilde{\zeta}}{\Sigma}-\hat{\Sigma}}\le \frac{\eta}{\dims+1}.
\]
Recall that $\eta = O(t(2\gamma)^{1-\frac{1}{2t}})$. Taking $t$ to be the smallest power of $2$ at least $\log(1/\gamma)$ proves that 
\[
\hsnorm{\expectDistrOf{\tilde{\zeta}}{\Sigma}-\hat{\Sigma}}=\tildeO{\gamma/\dims}.
\]
Using that for $\ns=\Omega(\dims^{\log(1/\gamma)}/\gamma^2)$, $\hsnorm{\hat{\Sigma}-\Sigma_\rho}=O(\gamma/\dims)$, we have our estimate $\expectDistrOf{\tilde{\zeta}}{\Sigma}$ is also close to the true $\Sigma_{\rho}$. Finally observe that $\rho = (\dims+1)\Sigma_\rho +\eye_\dims$ and $\hat{\rho}=(\dims+1)\hat{\Sigma} +\eye_\dims$ proves \cref{thm:tomography-alg}.

\subsection{Hypercontractivity of uniform POVM outcomes}
We first upper bound the covariance of $\qproj{v}$.

\begin{lemma}
    Let $\qbit{v}$ be drawn from the measure $\dims\matdotprod{v}{\rho}{v}\dm v$ where $\dm v$ is the Haar measure on the complex unit sphere. Then 
    \[
    \expectDistrOf{v\sim\cD(\rho)}{\qproj{v}^{\otimes2}}=\frac{1}{(\dims+1)(\dims+2)}{(\eye_{\dims}\otimes\eye_\dims+F)(\eye_\dims^{\otimes 2}+\eye_\dims \otimes\rho + \rho\otimes \eye_\dims)}\preceq \frac{6}{\dims^2}\eye_\dims\otimes\eye_\dims.
    \]
    \label{lem:unifor-povm-2-moment}
\end{lemma}

\begin{proof}
    We first compute the second moment of $\qproj{v}$ using \cref{equ:haar-k-moment}.
    \begin{align*}
        \expectDistrOf{v\sim\cD(\rho)}{\qproj{v}^{\otimes2}}&=\int\dims\qproj{v}^{\otimes 2}\matdotprod{v}{\rho}{v}\dm v\\
        &=\dims\Tr_1\left[\Paren{\int\qproj{v}^{\otimes3}\dm v}\rho\otimes \eye_{\dims}^{\otimes 2}\right]\\
        &=\dims\binom{\dims+2}{3}\Tr_1\left[P_{\text{Sym}}^{(k)}\rho\otimes \eye_{\dims}^{\otimes 2}\right]\\
        &=\frac{6}{(\dims+1)(\dims+2)}\Tr_1\left[P_{\text{Sym}}^{(k)}\rho\otimes \eye_{\dims}^{\otimes 2}\right].
    \end{align*}
Here $\Tr_1$ denotes partial trace over the first component, i.e. $\Tr_1[A\otimes B\otimes C]=\Tr[A]B\otimes C$.  

Next we compute $\Tr_1[P_\pi\rho\otimes \eye_\dims\otimes \eye_\dims]$ for every $\pi\in \Sim_3$. Let $\rho =\sum_{i=1}^\dims \lambda_i\qproj{\psi_i}$ where $\lambda_i\ge 0$, $\sum_{i=1}^\dims\lambda_i=1$, $\{\qbit{\psi_i}\}_{i=1}^\dims$ are orthonormal. 
\begin{align*}
    \Tr_1[P_{\pi^{-1}} \rho \otimes\eye_\dims \otimes\eye_\dims] &= \Tr\left[P_{\pi^{-1}} \sum_{i_1=1}^\dims\sum_{i_2=1}^\dims\sum_{i_3=1}^\dims \lambda_{i_1}\qproj{\psi_{i_1}}\otimes \qproj{\psi_{i_2}}\otimes\qproj{\psi_{i_3}}\right]\\
    &=\sum_{i_1=1}^\dims\sum_{i_2=1}^\dims\sum_{i_3=1}^\dims \lambda_{i_1}\Tr_1[\qbit{\psi_{i_{\pi(1)}}}\qadjoint{\psi_{i_1}}\otimes \qbit{\psi_{i_{\pi(2)}}}\qadjoint{\psi_{i_2}}\otimes\qbit{\psi_{i_{\pi(3)}}}\qadjoint{\psi_{i_3}}]\\
    &=\sum_{i_1=1}^\dims\sum_{i_2=1}^\dims\sum_{i_3=1}^\dims\lambda_{i_1}\indic{i_{\pi(1)}=i_1}\qbit{\psi_{i_{\pi(2)}}}\qadjoint{\psi_{i_2}}\otimes\qbit{\psi_{i_{\pi(3)}}}\qadjoint{\psi_{i_3}}
\end{align*}

Denote $F=P_{21}\in \C^{\dims^2\times\dims^2}$ as the flip operator, i.e. $F\qbit{x}\otimes\qbit{y}=\qbit{y}\otimes\qbit{x}$. We can evaluate the above expression for each $\pi \in \Sim_3$. Note that permutation operators over $(\C^{\dims})^{\otimes3}$ can be expressed using $F$ and $\eye_\dims$,
\begin{align*}
    P_{123}&=\eye_\dims^{\otimes3},\quad P_{213}=F\otimes\eye_\dims, \quad P_{132}=\eye_\dims \otimes F, \\
    P_{231} &= P_{213}P_{132}, \quad P_{312}=P_{132}P_{213},\\
    P_{321}&=P_{213}P_{132}P_{213}.
\end{align*}
Therefore,
\begin{align*}
    \Tr_1[P_{123}(\rho\otimes\eye_\dims^{\otimes2})]&=\Tr_1[\rho\otimes\eye_\dims^{\otimes2}]=\eye_\dims^{\otimes2}\\
    \Tr_1[P_{132}(\rho\otimes\eye_\dims^{\otimes2})]&=\Tr_1[(\eye_\dims \otimes F)(\rho\otimes \eye_\dims^{\otimes 2})]=F\\
    \Tr_1[P_{213}(\rho\otimes\eye_\dims^{\otimes2})]&=\sum_{i_1=1}^\dims\sum_{i_3=1}^\dims\lambda_{i_1}\qproj{\psi_{i_1}}\otimes\qbit{\psi_{i_{3}}}\qadjoint{\psi_{i_3}}=\rho\otimes\eye_\dims
\end{align*}
By symmetry, $\Tr_1[P_{321}(\rho\otimes\eye_\dims^{\otimes2})]=\eye_\dims \otimes \rho$. We are left with the single-cycle permutations,
\begin{align*}
    \Tr_1[P_{231}]&=\sum_{i_1=1}^\dims\sum_{i_2=1}^\dims\sum_{i_3=1}^\dims\lambda_{i_1}\indic{i_{3}=i_1}\qbit{\psi_{i_{1}}}\qadjoint{\psi_{i_2}}\otimes\qbit{\psi_{i_{2}}}\qadjoint{\psi_{i_3}}\\
    &=\sum_{i_1=1}^\dims\sum_{i_2=1}^\dims \lambda_{i_1}\qbit{\psi_{i_{1}}}\qadjoint{\psi_{i_2}}\otimes\qbit{\psi_{i_{2}}}\qadjoint{\psi_{i_1}}\\
    &= F(\eye_\dims\otimes \rho) = (\rho\otimes \eye_\dims)F.
\end{align*}
Likewise,
\[
\Tr_1[P_{312}]=F(\rho\otimes\eye_\dims)=(\eye_\dims \otimes \rho)F.
\]
Combining all the parts, 
\begin{align}
    \expectDistrOf{v\sim\cD(\rho)}{\qproj{v}^{\otimes2}}&=\frac{1}{(\dims+1)(\dims+2)}\Paren{\eye_{\dims}\otimes\eye_\dims+F+\rho\otimes \eye_\dims + \eye_\dims \otimes \rho + F(\eye_\dims \otimes\rho + \rho\otimes \eye_\dims)}\nonumber\\
    &=\frac{1}{(\dims+1)(\dims+2)}{(\eye_{\dims}\otimes\eye_\dims+F)(\eye_\dims^{\otimes 2}+\eye_\dims \otimes\rho + \rho\otimes \eye_\dims)}.\label{equ:unif-povm-2-moment}
\end{align}

As a sanity check, since $F(\eye_\dims \otimes\rho + \rho\otimes \eye_\dims)=(\eye_\dims \otimes\rho + \rho\otimes \eye_\dims)F$, so the above expression is a Hermitian matrix. Since $\rho\preceq \eye_\dims $, we have $0\preceq \eye_\dims^{\otimes 2}+\eye_\dims \otimes\rho + \rho\otimes \eye_\dims\preceq 3\eye_\dims ^{\otimes 2}$. Furthermore, $F$ is also Hermitian with $F^2=\eye_\dims^{\otimes 2}$, so $0\preceq\eye_\dims^{\otimes 2}+F\preceq 2\eye_\dims^{\otimes 2}$. Therefore we have 
\[
(\eye_{\dims}\otimes\eye_\dims+F)(\eye_\dims^{\otimes 2}+\eye_\dims \otimes\rho + \rho\otimes \eye_\dims)\preceq 6\eye_\dims^{\otimes2}.
\]
This is because if two Hermitian matrices commute, then they can be diagonalized using the same unitary matrix. Thus we prove the lemma.
\end{proof}
\begin{corollary}[Certifiably bounded second moment of observables]
    For all Hermitian matrix $M\in\Herm{\dims}$, the distribution $\cD(\rho)$ satisfies
    \[
    \expectDistrOf{v\sim \cD(\rho)}{\matdotprod{v}{M}{v}^2}\le \frac{4(\Tr[M^2]+\Tr[M]^2)}{(\dims+1)(\dims+2)},\quad \variance{v\sim \cD(\rho) }{\matdotprod{v}{M}{v}}\le \frac{3}{\dims^2}\hsnorm{M}^2.
    \]
    Moreover, there exists a degree-2 sum-of-squares proof in $M$ (treating the real and imaginary parts as indeterminate).
\end{corollary}
\begin{proof}
    We directly work on the closed-form expression from \cref{lem:unifor-povm-2-moment}. Using \eqref{equ:perm-trace}, we have
    \begin{align*}
        \expectDistrOf{v\sim \cD(\rho)}{\matdotprod{v}{M}{v}^2}&= \frac{1}{(\dims+1)(\dims+2)}\Paren{\Tr[M]^2+\Tr[M^2]+2\Tr[M]\Tr[M\rho]+2\Tr[M^2\rho]}.
    \end{align*}
    It suffices to prove that for all $M$
    \[
    3\Tr[M]^2+3\Tr[M^2]-2\Tr[M]\Tr[M\rho]-2\Tr[M^2\rho]\ge 0.
    \]
    Since it is a degree-2 homogeneous polynomial, it must be SoS. We consider the following matrix,
    \[
    P=(3\Tr[M]^2+3\Tr[M^2])\eye_\dims -2\Tr[M]M-2\Tr[M^2].
    \]
    We just need to prove it is p.s.d, then $\Tr[P\rho]\ge 0$ for all quantum state $\rho$. Let $M=\sum_{i=1}^\dims \lambda_i\qproj{u_i}$, then
    \begin{align*}
        P=\sum_{i=1}^\dims(3\Tr[M]^2+3\Tr[M^2]-2\Tr[M]\lambda_i-2\lambda_i^2)\qproj{u_i}.
    \end{align*}
    We further prove that each eigenvalue is non-negative.
    \begin{align*}
        3\Tr[M]^2+3\Tr[M^2]-2\Tr[M]\lambda_i-2\lambda_i^2=2\Tr[M]^2+ (\Tr[M]-\lambda_i)^2+3\Tr[M^2]-3\lambda_i^2\ge 0.
    \end{align*}
    In the final step, we used that $\Tr[M^2]=\sum_{i}\lambda_i^2$. This completes the proof of the first upper bound.
    
    Next we prove the variance bound, using \eqref{equ:sigma-rho},
    \[
    \expectDistrOf{v\sim \cD(\rho)}{\matdotprod{v}{M}{v}}=\frac{\Tr[M]+\Tr[M\rho]}{\dims+1}.
    \]
    Therefore,
    \begin{align*}
        \variance{v\sim \cD(\rho) }{\matdotprod{v}{M}{v}}&=\expectDistrOf{v\sim \cD(\rho)}{\matdotprod{v}{M}{v}^2}-\expectDistrOf{v\sim \cD(\rho)}{\matdotprod{v}{M}{v}}^2\\
        &=\frac{(\dims+1)(\Tr[M^2]+2\Tr[M^2\rho]-\Tr[M\rho]^2)-(\Tr[M]+\Tr[M\rho])^2}{(\dims+1)^2(\dims+2)}\\
        &\le \frac{\Tr[M^2]+2\Tr[M^2\rho]}{(\dims+1)(\dims+2)}\\
        &\le \frac{3\Tr[M^2]}{(\dims+1)(\dims+2)}.
    \end{align*}
    The final inequality is due to $0\preceq\rho\preceq\eye_\dims$, and thus $\Tr[M^2(\eye_\dims -\rho)]\ge 0$, which is a degree-2 homogeneous polynomial in $M$. This completes the proof.
\end{proof}

\begin{theorem}[Certifiable hypercontractivity of uniform POVM] For all integers $h\ge 2$, $\cD(\rho)$ satisfies $C$-hypercontractivity, i.e., there exists constant $C>0$ for all Hermitian matrix $M$, 
\[
\expectDistrOf{v\sim\cD(\rho)}{\matdotprod{v}{M}{v}^{h}}^2\le \frac{(Ch)^{2h}}{\dims^{2h}}(\Tr[M^2]+\Tr[M]^2)^h\le  (\sqrt{2}Ch)^{2h}{\expectDistrOf{v\sim \cD(\rho)}{\matdotprod{v}{M}{v}^2}}^{h}.
\]
Moreover, there exists a degree $2h$ sum-of-squares proof, i.e.,  RHS$-$LHS is a sum-of-squares polynomial in terms of the real and imaginary parts of $H$.
   \label{thm:hypercontractive} 
\end{theorem}

\begin{proof}

First we give a sum-of-squares lower bound for the second moment,
\begin{lemma}
Let $M\in\Herm{\dims}$ be the indeterminate. There is a degree-2 SoS proof of
    \[
    (\dims+1)(\dims+2)\expectDistrOf{v\sim \cD(\rho)}{\matdotprod{v}{M}{v}^2}\ge \frac{1}{2}(\Tr[M^2]+\Tr[M]^2).
    \]
    \label{lem:sos-2moment-lower}
\end{lemma}
\begin{proof}
    Recall the expression for the second moment of $\matdotprod{v}{M}{v}$,
    \begin{align*}
        &\quad (\dims+1)(\dims+2)\expectDistrOf{v\sim \cD(\rho)}{\matdotprod{v}{M}{v}^2}-\frac{1}{2}(\Tr[M^2]+\Tr[M]^2)\\
        &= \frac12\Paren{\Tr[M]^2+\Tr[M^2]+4\Tr[M]\Tr[M\rho]+4\Tr[M^2\rho]}\\
        &=\frac12\Paren{\Tr[M^2]^2+\Tr[(\Tr[M]^2\eye_\dims+4\Tr[M]M+4M^2)\rho] }.
    \end{align*}
    $\Tr[M^2]$ is already a degree-2 SoS. Thus it suffices to prove that $P\eqdef\Tr[M]\eye_\dims+4\Tr[M]M+4M^2$ is p.s.d. Let $M=\sum_{i=1}^\dims \lambda_i \qproj{u_i}$. Then,
    \[
    P=\sum_{i=1}^\dims (\Tr[M]^2+4\Tr[M]\lambda_i+4\lambda_i^2)\qproj{u_i}=\sum_{i=1}^\dims (\Tr[M]+2\lambda_i)^2\qproj{u_i}.
    \]
    Thus $P$ is p.s.d., and $\Tr[P\rho]\ge 0$ for all $M, \rho$. Note that $\Tr[P\rho]$ is a degree-2 homogeneous polynomial in $M$, so it must be a sum-of-squares polynomial. The proof is complete.
\end{proof}

Then it suffices to obtain an SoS upper bound of the $h$-order moment. Using \eqref{equ:haar-k-moment} ,
\begin{align*}
    \expectDistrOf{v\sim \cD(\rho)}{\matdotprod{v}{M}{v}^h}&=\binom{\dims + h }{h+1}^{-1}\frac{1}{h!}\sum_{\pi\in \Sim_{h+1}}\Tr\left[P_{\pi}(M^{\otimes h}\otimes \rho)\right]\\
    &=\frac{1}{(\dims+1)\cdots(\dims+h)}\sum_{\pi\in \Sim_{h+1}}\Tr\left[P_{\pi}(M^{\otimes h}\otimes \rho)\right].
\end{align*}
Using SoS almost triangle inequality,
\begin{equation}
    \sststile{2h}{M}\left\{
(d+1)^h\expectDistrOf{v\sim \cD(\rho)}{\matdotprod{v}{M}{v}^h}^2\le (h+1)! \sum_{\pi\in \Sim_{h+1}} \Tr\left[P_{\pi}(M^{\otimes h}\otimes \rho)\right]^2 \right\}.
\label{equ:sos-h-moment-triangle}
\end{equation}

Thus it suffices to prove that $\Tr\left[P_{\pi}(M^{\otimes h}\otimes \rho)\right]^2\le (\Tr[M^2]+\Tr[M]^2)^{h}$.
For permutations $\tau\in \Sim_k$ with only one cycle, by $\eqref{equ:perm-trace}$ and symmetry,
\[
\Tr[P_\tau M^{\otimes k}\otimes\rho ]=\Tr[M^k\rho].
\]
We decompose $P_\pi$ according to cycles. Let $c_{h+1}$ be the cycle that contains the element $h+1$ (which corresponds to $\rho$), then
\[
\Tr\left[P_{\pi}(M^{\otimes h}\otimes \rho)\right]^2=\Tr\left[M^{|c_{h+1}|}\rho\right]^2\prod_{c\ne c_{h+1}}\Tr[M^{|c|}]^2.
\]
We use the Lagrange's identity, which is a sum-of-squares proof for Cauch-Schwarz. For all complex numbers $a_1, \ldots, a_k, b_1, \ldots, b_k$,
\begin{equation}
    \Paren{\sum_{i=1}^k|a_i|^2}\Paren{\sum_{i=1}^k|b_i|^2}-\Paren{\sum_{i=1}^ka_i\overline{b_i}}^2=\sum_{i=1}^{k-1}\sum_{j=i+1}^k|a_ib_j-a_jb_i|^2.
\end{equation}
For all integers $k$, 
\[
\sststile{2k}{M}\left\{\Tr\left[M^{2k}\right]\le \Tr[M^2]^k\right\}.
\]
Therefore,
\[
\sststile{2|c_{h+1}|}{M}\left\{\Tr\left[M^{|c_{h+1}|}\rho\right]^2\le \Tr\left[M^{2|c_{h+1}|}\right]\Tr[\rho^2]\le\Tr\left[M^{2|c_{h+1}|}\right]\le \Tr[M^2]^{|c_{h+1}|}\right\}.
\]
In the penultimate step we used that $\Tr[\rho^2]\le 1$. In addition, for $|c|\ge 3$,
\[
\sststile{2|c|}{M}\left\{\Tr\left[M^{|c|}\right]^2\le \Tr\left[M^{2|c|-2}\right]\Tr[M^2]\le \Tr[M^2]^{|c|}\right\}.
\]
For $|c|=1$ and $2$, 
\[
\sststile{2|c|}{M}\left\{Tr[M^{|c|}]\le \Tr[M^2]+\Tr[M]^2\right\}.
\]
Combining all the parts using multiplication rule of SoS proof system, we have
\[
\sststile{2\sum_{c}|c|=2h}{M}\left\{\Tr\left[P_{\pi}(M^{\otimes h}\otimes \rho)\right]^2\le (\Tr[M^2]+\Tr[M]^2)^{h}\right\}.
\]
Using \eqref{equ:sos-h-moment-triangle} and the addition rule,
\[
    \sststile{2h}{M}\left\{
(d+1)^h\expectDistrOf{v\sim \cD(\rho)}{\matdotprod{v}{M}{v}^h}^2\le ((h+1)!)^2 (\Tr[M^2]+\Tr[M]^2)^{h}\right\}.
\]
Applying \cref{lem:sos-2moment-lower} proves the theorem with $C=2$.

\end{proof}

\begin{theorem}[Certifiable hypercontractivity under sampling]
    \label{thm:hypercontracivet-sample}
    Let $h\ge 2$ and $\cV=\{\qbit{v_1}, \ldots, \qbit{v_n}\}$ be i.i.d. samples from $\cD(\rho)$ and $\ns = \Omega(h\dims^{8h})$. Then with probability at least $1-1/\poly{\ns}$, the uniform distribution over the samples is $h$-certifiably $c$-contractive for some constant $c$, 
    \[
    \expectDistrOf{v\sim \cV}{\matdotprod{v}{M}{v}^h}^2\le  \left(\frac{ch}{\dims}\right)^{2h}(\Tr[M^2]+\Tr[M]^2)^{h}.
    \]
\end{theorem}
\begin{proof}
   The proof steps in \cite[Lemma 8.3]{bakshi2020outlierrobust} yields the following result,
   \begin{lemma}[From {\cite[Lemma 8.3]{bakshi2020outlierrobust}}]
   \label{lem:opnorm-samples}
       Let $\cD$ be 1-subgaussian in $\R^\dims$ and $\mathcal{S}$ be $n=\Omega((h\dims)^{8h}$ i.i.d. samples from $\cD$, then there exists constant $c$, with probability at least $1-1/\poly(\ns)$,
       \[
       \opnorm{\expectDistrOf{x\sim \mathcal{S}}{\qproj{x}^{\otimes h}}-\expectDistrOf{x\sim \cD}{\qproj{x}^{\otimes h}}}\le (ch)^h.
       \]
   \end{lemma}
   We can express complex vectors $\qbit{v}=\qbit{x}+\img \qbit{y}$ where $\qbit{x},\qbit{y}\in  \R^\dims$. Denote $\qbit{u}\in\R^{2\dims}$ as the concatenation of $\qbit{x},\qbit{y}$. Then we have $\qbit{v}=A \qbit{u}$ where $A=(\eye_\dims , \img\eye_\dims)$. 
   
   Note that $\cD(\rho)$ is $1\sqrt{\dims}$-subgaussian. Therefore, $\qbit{u}$ is also $1/\sqrt{\dims}$-subgaussian when $\qbit{v}\sim \cD(\rho)$.
   
   We are interested in the operator norm of 
   \[
   P\eqdef \expectDistrOf{v\sim \cV}{\qproj{v}^{\otimes h}}-\expectDistrOf{v\sim \cD(\rho)}{\qproj{v}^{\otimes h}}
   \]
   Note that $P=A^\dagger Q A$ where $Q=\expectDistrOf{v\sim \cV}{\qproj{u}^{\otimes h}}-\expectDistrOf{v\sim \cD(\rho)}{\qproj{u}^{\otimes h}}$, which by \cref{lem:opnorm-samples} has operator norm at most $(ch)^h/\dims^h$. Since $\opnorm{A}\le 2$, by sub-multiplicavity of operator norms, we have $\opnorm{P}\le 4\opnorm{Q}\le 4(ch)^h/\dims^h$. Therefore using SoS operator norm bound (see \cite[Eq (8.2)]{bakshi2020outlierrobust} )
   \begin{align*}
       \sststile{2h}{M}\left\{\Tr[M^{\otimes h}P]\le \frac{4(ch)^h}{\dims^h} \hsnorm{M}^h \right\}.
   \end{align*}
   Rearanging the terms using the definition of $P$,
   \begin{align*}
       \sststile{2h}{M}\left\{\Tr[M^{\otimes h} \expectDistrOf{v\sim \cV}{\qproj{v}^{\otimes h}}]\le \frac{4(ch)^h}{\dims^h} \hsnorm{M}^h +  \Tr[M^{\otimes h}\expectDistrOf{v\sim \cD(\rho)}{\qproj{v}^{\otimes h}}]\right\}.
   \end{align*}
   Using $h$-certifiable hypercontractivity of $\cD(\rho)$ proves the lemma.

\end{proof}

\section{Robust state testing algorithm}\label{sec:testing}
The algorithm idea is to first randomly sample a basis measurement and apply discrete distribution testing to the measurement outcomes. If the testing algorithm is adversarially robust, then so would be our quantum state testing algorithm.

\begin{theorem}[{\cite[Theorem 2]{acharya21manipulation}}]
\label{thm:robust-distr-testing}
Let $\q$ be a known distribution over a domain of size $\ab$. Given $\ns$ $\gamma$-corrupted samples $\bx=(x_1, \ldots, x_{\ns})$ from a discrete distribution $\p$ , there exists an algorithm \\\texttt{RobustDistrTesting}($\bx, \gamma, \q$) that can test whether $\norm{\p-\q}_1>\eps$ or $\p=\q$ with probability at least 0.9 and an accuracy of
\[
\eps=\bigO{\frac{\ab^{1/4}}{\sqrt{\ns}}+\gamma+\sqrt{\frac{\ab\gamma}{\ns}}+\Paren{\frac{\ab}{\ns}}^{1/4}\sqrt{\gamma}}.
\]
\end{theorem}

We now describe the robust quantum state testing algorithm in \cref{alg:robust-testing}.
\begin{algorithm}
\caption{Robust quantum state testing algorithm}
\label{alg:robust-testing}
    \begin{algorithmic}
        \State \textbf{Input}: $\ns$ copies of $\rho\in\C^{\dims\times\dims}$, 
        known state description $\qkn$, corruption parameter $\gamma$.
        \State \textbf{Output}: {\isthestate} if $\rho=\qkn$, {\notthestate} if $\tracenorm{\rho-\qkn}$ is large.
        \State Sample a Haar-random unitary $U=[u_1, \ldots, u_{\dims}]$. 
        \State Apply basis measurement $\POVM_U=\{\qproj{u_i}\}_{i=1}^{\dims}$ to all copies.
        \State Let $\p_{\rho}^{U}$ be the outcome distribution measuring $\rho$ with $\POVM_U$ (resp. $\p_{\qkn}^{U}$).
        \State Obtain $\gamma$-corrupted outcomes $\bx=(x_1, \ldots, x_{\ns})$.
        \State \Return \texttt{RobustDistrTesting}($\bx, \gamma, \p_{\qkn}^U$)
    \end{algorithmic}
\end{algorithm}

\begin{theorem}
    \label{thm:robust-testing-alg}
    Under $\gamma$-adversarial corruption, with $\ns$ copies of $\rho$, \cref{alg:robust-testing} can test if $\rho=\qkn$ or $\tracenorm{\rho-\qkn}>\eps$ with probability at least 0.8 and an accuracy of 
    \[
    \eps=\bigO{\frac{\dims^{3/4}}{\sqrt{\ns}}+\gamma\sqrt{\dims}+\frac{\dims^{3/4}\sqrt{\gamma}}{\ns^{1/4}}}\log\dims.
    \]
\end{theorem}

The error matches the lower bound up to log factors for $\gamma\le \dims^{1/4}/\sqrt{\ns}$ and  $\gamma\ge \sqrt{\dims/\ns}$. 
\begin{proof}
    The key is to argue that when $\rho$ and $\qkn$ are far apart, then the $\ell_1$ distance between the outcome distributions $\p_{\rho}^U,\p_{\qkn}^U$ should also be far with a high constant probability.  
    Formally, we have the following lemma,
    \begin{lemma}
    \label{lem:l1-output-distr}
        There exists a constant universal $c$ such that with probability at least 0.9 over a Haar random unitary $U$, 
        \[
        \norm{\p_{\rho}^U-\p_{\qkn}^U}_1\ge c\cdot\frac{\hsnorm{\rho-\qkn}}{\log\dims}.
        \]
    \end{lemma}
    The proof involves computing the Haar integral, which is deferred to \cref{sec:lem:l1-output-distr}. When $\tracenorm{\rho-\qkn}>\eps$, by Cauchy-Schwarz inequality we have $\hsnorm{\rho-\qkn}>\eps/\sqrt{\dims}$, and thus with probability at least 0.9 over a Haar-random $U$,
    \[
    \norm{\p_{\rho}^U-\p_{\qkn}^U}_1\ge c\cdot\frac{\eps}{\sqrt{\dims}\log\dims}.
    \]
    As a result, it suffices to run a distribution testing algorithm with $\p_{\qkn}^U$ as the reference distribution. Since the outcome domain size is $\dims$, by \cref{thm:robust-distr-testing}, with $\ns$ $\gamma$-corrupted samples, we can achieve an $\ell_1$ accuracy of 
    \[
    \eps'=\bigO{\frac{\dims^{1/4}}{\sqrt{\ns}}+\gamma+\sqrt{\frac{\dims\gamma}{\ns}}+\Paren{\frac{\dims}{\ns}}^{1/4}\sqrt{\gamma}}.
    \]
    Thus we can set $\eps'=\frac{\eps}{\sqrt{\dims}\log\dims}$ and obtain
    \[
    \eps=\bigO{\frac{\dims^{3/4}}{\sqrt{\ns}}+\gamma\sqrt{\dims}+\dims\sqrt{\frac{\gamma}{\ns}}+\frac{\dims^{3/4}\sqrt{\gamma}}{\ns^{1/4}}}\log\dims.
    \]
    The lower bound \cref{thm:robust-lower} suggests that we must have $\gamma\le 1/\sqrt{\dims}$ to achieve any meaningful testing result. 
    In this case, the third term above is always dominated by the first term, i.e. $\dims\sqrt{\frac{\gamma}{\ns}}\le \frac{\dims^{3/4}}{\ns}$, and thus can be removed. 
    
    Finally we bound the success probability. Conditioned on $\norm{\p_{\rho}^U-\p_{\qkn}^U}_1$ are sufficiently far apart, which by \cref{lem:l1-output-distr} happens with probability at least 0.9, the output of \texttt{RobustDistrTesting} is correct with probability at least 0.9. Thus, the success probability is at least $0.9\times 0.9>0.8$. This completes the proof of the theorem.
\end{proof}

\subsection{Proof of \cref{lem:l1-output-distr}}
\label{sec:lem:l1-output-distr}
We prove \cref{lem:l1-output-distr} in this section. Recall that a unitary $U=[u_1, \ldots, u_\dims]$ is sampled uniformly from the Haar measure, and $\POVM_U=\qproj{u_i}$ is the corresponding basis measurement. When the state is $\rho$, the measurement outcome follows distribution $\p_{\rho}^U$ where
\[
\p_{\rho}^U(x)=\matdotprod{u_x}{\rho}{u_x}, x=1,\ldots, \dims.
\]
The idea is to use H\"older's inequality,
\[
\norm{\p_{\rho}^U-\p_{\qkn}^U}_2^2\le \norm{\p_{\rho}^U-\p_{\qkn}^U}_1\norm{\p_{\rho}^U-\p_{\qkn}^U}_\infty.
\]
Thus, it suffices to lower bound the $\ell_2$ norm and upper bound the infinity norm. The lower bound on $\ell_2$ norm was proven in prior works~\cite{ChenLO22instance, liu2024restricted}. 
\begin{lemma}
For $\dims\ge 200$, with probability at least 0.95 over a Haar random unitary $U$, 
   \[
   \norm{\p_{\rho}^U-\p_{\qkn}^U}_2\ge 0.07 \frac{\hsnorm{\rho-\qkn}}{\sqrt{\dims}}.
   \]
   \label{lem:haar-l2-lower}
\end{lemma}
We refer the readers to \cite[Lemma 6.4]{ChenLO22instance} or \cite[Section C.2.2]{liu2024restricted} for detailed proof.

Next we upper bound the $\ell_\infty$ norm. We use the fact that $\ell_\infty$ norm is upper bounded by all $\ell_p$ norms and then choose the $p$ that minimizes the upper bound.

\begin{lemma}
    Let $\Delta=\rho-\qkn$ and $U$ sampled uniformly from the Haar measure. For all even integers $p\ge 2$,
    \[
    \expectDistrOf{U}{\norm{\p_{\rho}^U-\p_{\qkn}^U}_p^p}\le\frac
    {\dims}{2}\binom{\dims+p-1}{p}^{-1}\hsnorm{\Delta}^p.
    \]
    \label{lem:expect-lp-norm}
\end{lemma}
\begin{proof}
Note that for $x=1, \ldots, \dims$, 
\[
\p_{\rho}^U(x)-\p_{\qkn}^U(x)=\matdotprod{u_x}{\rho-\qkn}{u_x}=\matdotprod{u_x}{\Delta}{u_x}.
\]
Thus using \eqref{equ:haar-trace-moment},
\begin{align*}
    \expectDistrOf{U}{(\p_{\rho}^U(x)-\p_{\qkn}^U(x))^p}&=\expectDistrOf{U}{\matdotprod{u_x}{\Delta}{u_x}^p}\\
    &=\binom{\dims+p-1}{p}^{-1}\frac{1}{p!}\sum_{\pi\in\Sim_p}\prod_{c\in \cycle(\pi)}\Tr[\Delta^{|c|}]
\end{align*}
Since $\Tr[\Delta=0]$, the product would be 0 if $\pi$ has a cycle of size 1. Thus, we only have to sum over all $\pi$ that is a derangement, i.e., there does not exist $i\in[p]$ such that $\pi(i)=i$. For these permutations, all cycles must have length at least 2, and thus, by the monotinicity of Schatten norms,
\[
|\Tr[\Delta^{|c|}]|\le \Tr[|\Delta|^{|c|}]=\norm{\Delta}_{S_{|c|}}^{|c|}\le\hsnorm{\Delta}^{|c|}.
\]
Here $|\Delta|$ is the unique p.s.d square root of $\Delta$. The first inequality is because if $\Delta$ has eigenvalues $\lambda_1,\ldots, \lambda_\dims$, then the eigenvalues of $|\Delta|$ would be $|\lambda_1|,\ldots, |\lambda_\dims|$. Therefore,
\begin{align*}
    \expectDistrOf{U}{(\p_{\rho}^U(x)-\p_{\qkn}^U(x))^p}&\le \binom{\dims+p-1}{p}^{-1}\frac{1}{p!}\sum_{\pi\text{ derangement}}\prod_{c\in \cycle(\pi)}\hsnorm{\Delta}^{|c|}\\
&=\binom{\dims+p-1}{p}^{-1}\frac{1}{p!}\sum_{\pi\text{ derangement}}\hsnorm{\Delta}^{p}\\
&=\binom{\dims+p-1}{p}^{-1}\frac{1}{2}\hsnorm{\Delta}^{p}
\end{align*}
In the final step, we used the fact that the number of derangements in $\Sim_p$ is $p!\sum_{i=0}^p\frac{(-1)^i}{i!}\le \frac{p!}{2}$. Finally, we prove the lemma by summing over all $x\in[\dims]$ and using the linearity of expectation.
\end{proof}
\begin{corollary}
    For $\dims\ge 200$, with probability at least 0.95 over the Haar measure $\Haar{\dims}$,
    \[
    \norm{\p_{\rho}^U-\p_{\qkn}^U}_\infty \le \frac{3e\ln\dims}{\dims}.
    \]
    \label{cor:infinity-norm-upper}
\end{corollary}
\begin{proof}
    We choose $p=2\lfloor\ln \dims\rfloor$ in \cref{lem:expect-lp-norm}. By Markov's inequality, 
    \[
    \probaDistrOf{U}{\norm{\p_{\rho}^U-\p_{\qkn}^U}_p^p\ge 10\dims\binom{\dims+p-1}{p}^{-1}\hsnorm{\Delta}^p}\le 0.05.
    \]
    Thus, with probability at least 0.95,
    \begin{align*}
        \norm{\p_{\rho}^U-\p_{\qkn}^U}_\infty&\le \norm{\p_{\rho}^U-\p_{\qkn}^U}_p\\
        &\le\frac{(10\dims)^{1/p}\hsnorm{\Delta}}{\binom{\dims+p-1}{p}^{1/p}}\\
        &\le \frac{p(10\dims)^{1/p}\hsnorm{\Delta}}{\dims+p-1}\\
        &\le \frac{2\ln \dims\cdot 1.5e}{\dims}\hsnorm{\Delta}.
    \end{align*}
    The penultimate step follows by ${n\choose k}\ge (n/k)^k$. The final step is because $d^{1/\ln d}=e$ and $\dims \ge 200$. Thus, we prove the lemma.
\end{proof}

Combining \cref{lem:haar-l2-lower}, \cref{cor:infinity-norm-upper}, and H\"older's inequality, with probability at least 0.9,
\[
 \norm{\p_{\rho}^U-\p_{\qkn}^U}_1\ge \frac{\norm{\p_{\rho}^U-\p_{\qkn}^U}_2^2}{\norm{\p_{\rho}^U-\p_{\qkn}^U}_\infty}\ge \frac{0.0049}{3e}\cdot \frac{\hsnorm{\Delta}}{\ln \dims}.
\]
This proves \cref{lem:l1-output-distr}.

\section*{Acknowledgement}
Vladimir Braverman is supported by NSF 2528780 and the Naval Research (ONR) grant N00014-23-1-2737. Nai-Hui Chia is supported by NSF Awards FET-2243659 and FET-2339116 (CAREER) and DOE Quantum Testbed Finder Award DE-SC0024301. Yuhan Liu is supported by a Rice University Chairman postdoctoral fellowship and the Naval Research (ONR) grant N00014-23-1-2737.

\bibliography{refs}
\bibliographystyle{alpha}

\end{document}